\documentclass[a4paper,onecolumn,11pt,accepted=2024-12-12]{quantumarticle}
\pdfoutput=1
\usepackage[utf8]{inputenc}
\usepackage[english]{babel}
\usepackage[T1]{fontenc}
\usepackage{amsmath}
\usepackage{hyperref}

\usepackage[numbers,sort&compress]{natbib}

\usepackage{tikz}
\usepackage{lipsum}
\usepackage{amssymb}
\usepackage{mathcomp}
\usepackage{dsfont}
\usepackage{enumerate}
\usepackage{amsfonts}
\usepackage{amsthm,bbm}
\usepackage{enumitem}
\usepackage{verbatim}

\newtheorem{definition}{Definition}[section]
\newtheorem{lemma}[definition]{Lemma}

\newtheorem{theorem}[definition]{Theorem}
\newtheorem{corollary}[definition]{Corollary}

\theoremstyle{definition}
\newtheorem{remark}[definition]{Remark}

\numberwithin{equation}{section}

\def\cN{\mathcal{N}}
\def\a0{{\rm a}_0}

\def\Re{\mathrm{Re}}
\def\Im{\mathrm{Im}}

\title[Out-of-time-ordered correlators of mean-field bosons via Bogoliubov theory]{Out-of-time-ordered correlators of mean-field bosons via Bogoliubov theory}

\author{Marius Lemm}
\address{Department of Mathematics, University of T\"ubingen, Auf der Morgenstelle 10, 72076 T\"ubingen, Germany}
\email{marius.lemm@uni-tuebingen.de}
\author{Simone Rademacher}
\address{Department of Mathematics, Ludwig-Maximilians-Universit\"at M\"unchen, Theresienstr. 39, 80333 Munich, Germany}
\email{simone.rademacher@math.lmu.de}

\begin{document}

\begin{abstract}
Quantum many-body chaos concerns the scrambling of quantum information among large numbers of degrees of freedom. It rests on the prediction that out-of-time-ordered correlators  (OTOCs)  of the form $\langle [A(t),B]^2\rangle$ can be connected to classical symplectic dynamics. We rigorously prove a variant of this correspondence principle for mean-field bosons. We show that the $N\to\infty$ limit of the OTOC $\langle [A(t),B]^2\rangle$ is explicitly given by a suitable symplectic Bogoliubov dynamics. In practical terms, we describe the dynamical build-up of many-body entanglement between a particle and the whole system by an explicit nonlinear PDE on $L^2(\mathbb{R}^3) \oplus L^2(\mathbb{R}^3)$. For higher-order correlators, we   obtain an out-of-time-ordered analog of the Wick rule. The proof uses Bogoliubov theory. Our finding spotlights a new problem in nonlinear dispersive PDE with implications for quantum many-body chaos.
\end{abstract}

\maketitle 

\section{Introduction}
The central objective of quantum chaos is to uncover connections between quantum dynamics and regular or chaotic features of the underlying classical dynamics. The subject took off in the 1980s and produced several success stories in physics and mathematics, e.g., the Gutzwiller 
trace formula and Berry's random wave model \cite{Haake,St}.

In 2008, motivated by the black hole information paradox, Sekino and Susskind \cite{SeSu} introduced the concept of ``fast information scramblers'', which is a notion of \textit{quantum many-body chaotic system}. In an influential talk in 2014, Kitaev \cite{Kitaev} identified a new link between black hole physics and quantum spin glasses. Kitaev proposed to capture information scrambling in quantum many-body systems through the notion of \textit{out-of-time-ordered-correlators} (which were originally introduced by Larkin and Ovchinnikov \cite{LO} in the context of superconductivity). The standard OTOC between two observables $A$ and $B$ is
\[
C_{A,B}(t)=-\langle\psi,  [A(t),B]^2\psi\rangle,\qquad \textnormal{where } A(t)=e^{\mathrm{i}tH} Ae^{-\mathrm{i}tH}.
\]
Here $A$ and $B$ are Hermitian operators such that initially $[A,B]=0$ and $\psi$ is a fixed reference state (e.g., an equilibrium state). For example, consider a bipartition of the Hilbert spaces $\mathcal H=\mathcal H_A\otimes \mathcal H_B$ and observables $A=A'\otimes \mathbbm 1_B$ and  $B=\mathbbm 1_A\otimes B'$. Then $[A,B]=0$, but $[A(t),B]\neq0$ for $t>0$, assuming the Hamiltonian $H$ couples $\mathcal H_A$ and $\mathcal H_B$. The OTOC $C_{A,B}(t)\geq 0$ measures the size of the commutator with respect to a fixed reference state.

In the past 10 years, OTOCs have emerged as the premier measure of entanglement growth in quantum many-body systems. They are being investigated and used by physicists from a variety of research areas, ranging from high-energy physics to condensed-matter physics and quantum information theory \cite{AFI,Blake,CZHF,FZSZ,FS,Hartnoll,HMY,HQRY,KH,KGP,LM,LS,MSS,NSZ,RSS,SZFZ,ShSt1,ShSt2,Stanford,Swingle,SBSH,SLSMD,YCZ, ZJL}. In particular, exponential growth of OTOCs on intermediate time scales ($C_{A,B}(t)\sim e^{v_{\mathrm{B}} t}$ where $v_{\mathrm{B}}\geq 0$ is the so-called butterfly velocity) is a key quantifier of quantum many-body chaos. By contrast, slow growth of OTOCs in time (say, logarithmic) is associated with many-body localization \cite{CZHF,FZSZ}. New experiments have been devised to directly measure OTOCs \cite{Betal,Getal,Mietal} (see also \cite{SBSH}). The Loschmidt echo \cite{GPSZ,YCZ} is a close relative which has a long history including on the experimental side \cite{AKD}. For additional background on OTOCs, see the reviews \cite{LSHOH,XS}. As emphasized in many of the references above, understanding OTOCs is one of the central problems in modern physics because it directly relates to understanding the mechanisms underpinning quantum many-body chaos and many-body localization.

What is the connection of the OTOC to classical physics? Essentially, since it is a commutator, a natural classical counterpart of the OTOC is the Poisson bracket.  For example, in a semiclassical situation in few-body quantum mechanics where $A$ and $B$ are pseudodifferential operators with symbols $a$ and $b$, $C_{A,B}(t)$ should be roughly given by $(\{a(t),b\})^2$ for times sufficiently short compared to the Ehrenfest time, where $a(t)$ is given in terms of the classical symplectic flow. The correspondence connects the notion of classical chaos (i.e., positive Lyapunov exponent) to rapid information scrambling  (i.e., exponential growth of OTOC $C_{A,B}(t)\sim e^{v_{\mathrm{B}} t}$). To see the heuristics, consider a single quantum particle living on a domain with chaotic geodesic flow and take the observables $A=x$, the position operator, and $B=-i\nabla_x$, the momentum operator. Then, assuming the correspondence between OTOC and Poisson bracket
\begin{equation}\label{eq:Poisson}
    C_{A,B}(t)\approx (\{x(t),-i\nabla_x\})^2=(\nabla_x x(t))^2\approx e^{2\lambda t},
\end{equation}
where we used that the Lyapunov exponent $\lambda$ quantifies the exponential divergence of trajectories $x(t)$.
The butterfly velocity $v_\mathrm{B}$ of information scrambling is then given by twice the classical Lyapunov exponent, $v_\mathrm{B}=2\lambda$. This semiclassical picture has been considered in low-dimensional quantum systems, e.g., the quantum kicked-rotor or kicked-top model, \cite{ASAND,JGW,RGG}.
Given the central role of the OTOC in quantum \textit{many-body} physics, it is an important problem to make an analogous correspondence principle between the OTOC and a symplectic counterpart rigorous in a many-body setting.\\

\subsection{Hamiltonian.}
In the present paper, we rigorously prove such a correspondence principle between OTOC and symplectic dynamics for a paradigmatic many-body system of condensed-matter physics: \textit{mean-field bosons}. That is, we consider a system of $N$ bosons in $\mathbb R^3$ in the mean-field regime, which is described by the Hamiltonian 
\begin{align}\label{eq:H}
H_N = \sum_{j=1}^N ( - \Delta_j ) + \frac{1}{N} \sum_{1\leq i<j\leq N} v(x_i-x_j) 
\end{align}
acting on the permutation-symmetric (i.e., bosonic) $N$-particle Hilbert space $L_s^2( \mathbb{R}^{3N})$. We assume that the two-body interaction potential $v:\mathbb R^3\to \mathbb R$ satisfies a weak regularity assumption
\begin{align}
\label{ass:v}
v^2 \leq C ( 1- \Delta)  \; 
\end{align}
which allows for singularities (e.g., the singular Coulomb interaction $v(x)=\frac{C}{|x|}$ in three dimensions by Hardy's inequality).  The mean-field regime is characterized by the prefactor $\frac{1}{N}$ in \eqref{eq:H} which corresponds to a special kind of weak coupling scaling limit.\footnote{Our model has some similarity with a many-body model considered in \cite{Stanford} by Stanford, see also \cite{LSHOH}, who identified the weak-coupling limit of the OTOC through a non-rigorous perturbative expansion.  
 The mean-field model being a scaling limit variant of weak coupling makes our model substantially easier to analyze and allows us to rigorously identify the large-$N$ limit of the OTOC by completely different mathematical methods.} 

 Standard arguments show that the (unbounded) Hamiltonian $H_N$ is self-adjoint on a suitable dense domain and so $e^{-\mathrm{i}t H_N}:L_s^2( \mathbb{R}^{3N})\to L_s^2( \mathbb{R}^{3N})$ is a unitary operator. This allows us to define the usual Heisenberg dynamics as 
\[
A(t)=e^{\mathrm{i}t H_N}Ae^{-\mathrm{i}t H_N}
\]
for suitable operators $A$ (the precise conditions on $A$ are given later). 

We will study the quantum many-body dynamics generated from $H$ in the usual way through the $N$-particle Schr\"odinger equation
\begin{align}
\label{eq:Schroe}
i \partial_t \psi_{N,t} = e^{-i t H_N} \psi_{N,t}
\end{align}

\subsection{Initial state.}
For the initial state, we are interested in a pure Bose-Einstein condensate state, a product state 
\begin{equation}\label{eq:psidefn}
\psi_N(x_1,\ldots,x_N) = \prod_{j=1}^N\varphi_0(x_j)
\end{equation}
for a one-body wave function $\varphi_0\in L^2(\mathbb R^3)$. We recall that trapped Bose systems at low temperatures are known to relax to the ground state and exhibit Bose-Einstein condensation experimentally \cite{Ketterle}.  The initial state \eqref{eq:psidefn} can thus be experimentally realized by trapping the bosons in an electric field, tuning the interaction strength down and cooling the system.   

Releasing the trap means that one considers the many-body time evolution \eqref{eq:Schroe} of the initial product state \eqref{eq:psidefn}. Our results will depend on the initial state displaying a BEC. The occurrence of BEC allows for an approximate, mathematically rigorous complexity reduction of the interacting many-body system. Heuristically, this is because, (i), the Bose-Einstein condensate  is to leading order described by the order parameter $\varphi_t$ and, (ii), the fluctuations around the condensate can also be described to leading order in terms of $\varphi_t$ by Bogoliubov theory. This twofold complexity reduction is what allows us to describe many-body information scrambling (at least to leading order) by one-body quantities. If the initial state is not close to a BEC, no meaningful complexity reduction of this kind is known and so there is no starting point for such an analysis. So the key physical mechanism underlying our result about describing the quantum many-body chaos of mean-field bosons is Bose-Einstein condensation.

\subsection{Description of the condensate through the nonlinear Hartree equation}
 A key fact is that the BEC is dynamically preserved to leading order. The condensate order parameter is a function $\varphi_t\in L^2(\mathbb R^3)$, called the condensate wave function. Its evolution is described by the nonlinear Hartree equation on $L^2(\mathbb R^3)$:
\begin{align}
\label{def:hartree}
i \partial_t \varphi_t = h_{\varphi_t} \varphi_t, \quad h_{\varphi_t} = - \Delta + (v*\vert \varphi_t \vert^2) \; . 
\end{align}
 The Hartree equation \eqref{def:hartree} is also known as the nonlinear Schr\"odinger equation (NLS), a prototypical nonlinear dispersive PDE that has entire monographs written on it \cite{Fibich}. It is of relevance not just for the description of BECs \cite{Ketterle}, but also for light propagation through nonlinear optical fibers and planar waveguides \cite{NLSreview}.

Indeed, the mathematical physics community has rigorously proved that the many-body dynamics \eqref{eq:psidefn} can be approximated by solution to the nonlinear Hartree equation \eqref{def:hartree} in a suitable sense; see for example \cite{AGT,AFP,BGM,CLS,C,EY,FKP,FKS,GV,KP,RoS,Sp}). Moreover the quantum fluctuations around the Hartree equation can be described \cite{BPPS,BS,C,GMM,H,LNS,MPP} in both dynamical and stationary settings by rigorous versions of Bogoliubov theory \cite{Bog}. Such boson systems have also been characterized in probabilistic language \cite{BKS,BSS,KRS,R,RS,Rsing}. In this paper, we shall draw on these techniques and use Bogoliubov theory to describe \textit{dynamical entanglement generation} via OTOCs between one boson and the whole system.

 The solution $\varphi_t$ to \eqref{def:hartree} is an important object for us, because our results will express the many-body OTOC (and its higher-order variants) to leading order in $N$ as commutators depending only on $\varphi_t$ (but not on $N$). The point is that the nonlinear Hartree equation is an equation on $\mathbb R^3$ and so it does not suffer from the ``curse of dimensionality'' inherent in many-body problems. In this sense, the question of many-body entanglement generation and information scrambling is reduced to studying $\varphi_t$. In particular, it is possible to study $\varphi_t$ numerically using standard routines for nonlinear equations. 
 
We conclude this discussion of the Hartree equation \eqref{def:hartree} with a  \textit{qualitative} description   of its solutions $\varphi_t$.

The Hartree equation (1.6) is a dispersive nonlinear PDE on $\mathbb R^3$. One source of the dispersion is the Laplacian. The qualitative effect of the nonlinearity depends  mainly on the sign of $v$.

 \begin{itemize}
     \item \textit{Defocusing case.} The more common situation is when the potential is purely repulsive $v(x)\geq 0$. In this case, the repulsive nature of the nonlinearity $v*|\varphi_t|^2$ tends to further increase the dispersion, because it penalizes accumulation. Accordingly, this is called the defocusing case. This effect is more pronounced the more local $v$ is, because locality penalizes accumulation more strongly. 
 In the defocusing case, the specific choice of $v$ affects the behavior of solutions quantitatively, but not qualitatively, because the main feature is dispersion. 
    \item \textit{Focusing case.} When  $v(x)\leq 0$, the potential is purely attractive. This has profound effects on the behavior of solutions as it produces a tendency to accumulate. This  ``focusing'' effect counteracts the dispersion. This can  lead to the existence of solitons, which are stationary solutions of the Hartree equation for which the dispersion and the attractive nonlinearity balance exactly to create a nonlinear version of a bound state. 
     \item \textit{General case.} For mixed-sign potentials, the qualitative behavior of solutions can be more subtle. A rough rule-of-thumb is that the sign of $\int_{\mathbb R^3} v(x) dx$ determines in which of the two cases (``defocusing'' versus ``focusing'') the behavior lies.
 \end{itemize}
 
 To summarize, the sign of $v$ affects the behavior of solutions qualitatively in a strong way.  The physically most common case is when $v\geq 0$ and so all solutions display a relatively strong form of dispersion, irrespective of the details of $v$. A prototypical case that is usually studied numerically is that of the power nonlinearity $\pm \lambda |\varphi_t|^2$ \cite{Fibich}.\footnote{While the choice $v(x-y)=\pm \lambda \delta(x-y)$ is technically not included in our statement because we assume that $v$ is a function, it could be easily included by studying a many body Hamiltonian with interaction $\frac{\lambda}{N}\sum_{i,j}v_N(x_i-x_j)$ with $v_N(x)=N^{3\epsilon} v(N^\epsilon x)$ for a fixed, small $\epsilon>0$ and $\int_{\mathbb R^3} vdx=1$. Indeed, $v_N(x)$ converges weakly to $\delta(x)$ as $N\to\infty$. We elected not to include this case to keep the presentation simple.}

\subsection{Observables and summary of main results}
An important difference to quantum spin systems, in which OTOCs have been studied a lot in the past, is that our degrees of freedom, the bosons, are free to move. Therefore, instead of studying entanglement generation (OTOCs) between different spatial \textit{regions}, it is more meaningful to investigate the dynamical generation of many-body entanglement (via OTOCs) between different \textit{bosons}. The kinds of observables whose OTOC we will study are defined in \eqref{eq:Bdefn} below.

We study the OTOC of the following observables. It is important to center the operators suitably and so we introduce
\begin{equation}\label{eq:Bdefn}
A_t^{(1)} = e^{itH_N} \left( A^{(1)} - \langle \varphi_t, A \varphi_t \rangle \right) e^{-itH_N},\qquad 
\mathcal{B} := \sum_{j=1}^N \Big( B^{(j)} - \langle \varphi_0, B \varphi_0 \rangle \Big) 
\end{equation}
where $A^{(j)}$ means that $A$ acts on the $j$ particle only. In words, $A_t^{(1)}$ captures the action of the Heisenberg dynamics on the first particle, while $\mathcal B$ is a reference operator that describes the collection of all particles. Therefore, the OTOC\footnote{Here we use the convention that the OTOC refers to the square of the commutator, which is a common definition in the literature, but not the only one.} that we study,
 $\langle \psi_N, [A_t^{(1)},\mathcal B_0]^2\psi_N\rangle $, indeed captures the generation on entanglement between the first particle and the whole system.

Then our first main result, \textbf{Theorem \ref{thm:OTOC}} identifies the $N\to\infty$ limit of the corresponding OTOC by an explicit nonlinear symplectic dynamics.
\begin{equation}\label{eq:main}
\boxed{-\lim_{N\to\infty}\langle \psi_N, [A_t^{(1)},\mathcal B_0]^2\psi_N\rangle 
= \langle (\varphi_0,J\varphi_0) [\Theta(t;0)A\Theta(t;0)^{-1},B](\varphi_0,J\varphi_0)\rangle_S}.
\end{equation}
Here $J:L^2(\mathbb R^3)\to  L^2(\mathbb R^3)$ is complex conjugation and the scalar product is given by
\begin{align}
\label{def:scalar-S}
\langle (f_1,f_2), (g_1,g_2) \rangle_{{\rm S}} := \langle (f_1,f_2), S (g_1,g_2) \rangle_{L^2 \oplus L^2},\qquad \text{with} \quad S = \begin{pmatrix}
1 & 0 \\
0 & -1
\end{pmatrix}.
\end{align}
The map $\Theta(t;0)$ is an isometry with respect to this scalar product, i.e., $S^*\Theta(t;0) S=\Theta(t;0)$, and it has symplectic structure,
\begin{align}
\label{eq:prop-Theta,J}
[\Theta (t;0), \mathcal{J}]
=0, \quad \text{with}\quad \mathcal{J} = \begin{pmatrix}
0 & J \\
J &0 
\end{pmatrix} \; . 
\end{align}
The map $\Theta(t;0)$ can be described explicitly as explained in the next section. It is the solution of a nonlinear dispersive PDE, essentially the Hartree equation $i\partial_t \varphi_t = -\Delta\varphi_t+ (v*|\varphi_t|^2)\varphi_t$ augmented with additional terms coming from Bogoliubov theory that generate the symplectic structure. As applications of \eqref{eq:main}, we prove that the information scrambling is initially ballistic (linear in time), as expected for generic quantum many-body systems (see also \cite{KH}). In fact, we obtain an explicit expression in terms of $\varphi_0$ for the initial velocity of information scrambling. We also derive an upper bound on the butterfly velocity which only depends on the initial data $\varphi_0$ (\textbf{Corollary \ref{cor:OTOC}}). Determining conditions such that the right-hand side in \eqref{eq:main} does indeed grow exponentially in time would allow to characterize conditions for quantum many-body chaos. Studying exponential growth of \eqref{eq:main} is an intriguing new kind of problem in nonlinear dispersive PDE that we would like to advertise here. (As a prior step, it would also be interesting to probe the growth properties on the right-hand side of \eqref{eq:main} numerically.)

In \textbf{Theorem \ref{thm:MCLT}} we derive the asymptotics of higher-order correlation functions at arbitrary times. This result requires sufficiently regular test functions and extends work of Buchholz-Saffirio-Schlein \cite{BSS} who considered the case of equal times. It raises the interesting possibility to identify a non-commutative Gaussian limit process in the framework of non-commutative probability (say, a variant of $q$-Gaussian processes \cite{BKSq}). 

In \textbf{Theorem \ref{thm:corrfct}}, we generalize the proof of Theorem \ref{thm:OTOC} to higher-order correlators between symmetrized creation and annihilation operators at arbitrary times. The result is a kind of out-of-time-ordered Wick rule. As in the proof of Theorem \ref{thm:OTOC}, we can avoid the sufficiently regular test functions needed for Theorem \ref{thm:MCLT} by a more direct analysis of second-quantized operators via Bogoliubov theory. 

We close by emphasizing that, at a methodological level, all of our proofs strongly build on the refined and mathematically rigorous rendition of Bogoliubov's theory  \cite{Bog} developed in the mathematical physics literature in the last years, e.g., \cite{BPPS,C,GMM,H,LNS,MPP}; see also  \cite{BKS,BSS,KRS,R,RS,Rsing}. While most prior works focused on a single fixed time, we find that  the techniques can be adapted to treat any finite collection of times simultaneously, which allows to treat OTOCs. This observation forges a novel connection between the mathematical description of mean-field bosons and the physics of many-body information scrambling expressed through OTOCs.

\begin{remark}[Connection to Lieb-Robinson bounds]
Lieb-Robinson bounds \cite{LR} are also bounds on $[A(t),B]$ but they are usually worst-case type bounds which are not expected to have any semiclassical content. They are mostly studied for quantum spin systems, which have bounded interactions. In this case, the Lieb-Robinson bound is of the form $\| [A(t),B]\|\leq Ce^{v_{\mathrm{LR}}t-d_{AB}}$, where the so-called Lieb-Robinson velocity $v_{\mathrm{LR}}$ depends on the operator norm of the interaction. Hence, in quantum spin systems, the LR velocity provides an upper bound on the butterfly velocity of entanglement scrambling, but this bound is typically not sharp. For bosons, Lieb-Robinson bounds are much more subtle because operator norms have to be avoided. For Bose-Hubbard type models, Lieb-Robinson bounds and related propagations bounds have been proved with respect to various classes of initial states \cite{FLS22a,FLS22b,KS21b,KVS22,LRSZ23,LRZ,SHOH,SZ,VKS,YL22}. For mean-field bosons with bounded interactions $v$, Lieb-Robinson bounds were proved by Erd{\H o}s-Schlein in \cite{ES}; these have $v_{\mathrm{LR}}=4\|v\|_\infty$.
\end{remark}


\section{Setup and main results}
\label{sect:setup}

\subsection*{Out-of-time-order correlators (OTOCs)} We are interested in the Heisenberg picture, i.e., the time evolution of observables. We use as reference state the pure Bose condensate $\psi_N = \varphi_0^{\otimes N}$. One of the observables is constructed from a one-particle operator $A$ on $L^2( \mathbb{R}^3)$ and we define for $j  \in \lbrace 1, \dots, N \rbrace$ the $N$-particle operator 
\begin{align}
A^{(i)} = \mathds{1} \otimes \cdots \otimes \mathds{1} \otimes A \otimes \mathds{1} \otimes \cdots \otimes \mathds{1} 
\end{align}
acting as $A$ on the $j$-th and as identity on the remaining $N-1$ particles. We consider the time evolved and centered (with respect to the condensates expectation value) observable $A_t^{(j)}$ given by 
\begin{align}
\label{def:Atj}
A_t^{(j)} = e^{itH_N} \left( A^{(j)} - \langle \varphi_t, A \varphi_t \rangle \right) e^{-itH_N} \;. 
\end{align}
For a self-adjoint operator $B$ and the sum 
\begin{align}
\label{def:B0}
\mathcal{B}_0 = \sum_{j=1}^N \Big( B^{(j)} - \langle \varphi_0, B \varphi_0 \rangle \Big) \; . 
\end{align}
the commutator of $A_0^{(1)}$ and $\mathcal{B}_0$ is at time $t=0$ given by 
\begin{align}
\langle \psi_N, [A^{(1)}_0 , \mathcal{B}_0] \psi_N \rangle  = \langle \psi_N, [A^{(1)}_0 , B^{(1)} ]\psi_N \rangle = \langle\varphi_0, [A,B] \varphi_0 \rangle = 2 \Im \langle \varphi_0, AB \varphi_0 \rangle \; . 
\end{align}
For positive times the particle's quantum correlations affect the commutators structure. We give a quantitative description of the so-called out-of-time-ordered-correlator that is related to the many-body Hamiltonian's quadratic Bogoliubov approximation. 

\subsection*{Bogoliubov dynamics} To state our main result we introduce the notation of a so-called Bogoliubov dynamics. For this we introduce the following operator kernels. 
\begin{align}
\label{def:tilde-Kj}
\widetilde{K}_{1,s} = q_s K_{1,s} q_s, \quad \widetilde{K}_{2,s} = (J q_s J)  K_{2,s} q_s 
\end{align}
with $q_s = 1 - \vert \varphi_s \rangle \langle \varphi_s \vert$ and $K_{j,s}$ the operators defined by their integral kernels
\begin{align}
\label{def:Kj}
K_{1,s} (x,y) = v(x-y) \varphi_s (x) \overline{\varphi}_s (y), \quad K_{2,s}(x,y) = v(x-y) \varphi_s (x) \varphi_s (y) \; . 
\end{align}
For us, the Bogoliubov dynamics is a two-parameter family of Bogoliubov transformations, each a bounded linear symplectic operator
\[
\Theta (t;s): L^2( \mathbb{R}^3) \oplus L^2( \mathbb{R}^3) \rightarrow L^2( \mathbb{R}^3) \oplus L^2( \mathbb{R}^3)
\]
 that is given in terms of two bounded linear maps $U(t;s), V(t;s) : L^2( \mathbb{R}^3) \rightarrow L^2( \mathbb{R}^3)$ by 
\begin{align}
\label{def:Theta}
\Theta (t;s) = \begin{pmatrix}
 U(t;s) & JV(t;s) J \\
 V (t;s) & JU(t;s) J
\end{pmatrix}
\end{align}
with $U(t;t)= 1, V(t;t) =0$. The Bogoliubov dynamics is characterized by the fact that $\Theta (t;s)$ satisfies, for all $0 \leq s \leq t$,
\begin{align}
\label{eq:Theta-s}
i \partial_s \Theta (t;s) = \mathcal{T}_s \Theta (t;s) \quad \text{with} \quad \mathcal{T}_s = \begin{pmatrix}
h_{\varphi_s} + \widetilde{K}_{1,s} &- \widetilde{K}_{2,s}  \\
\overline{\widetilde{K}}_{2,s} & - h_{\varphi_s} - \widetilde{K}_{1,s} \; 
\end{pmatrix} \;
\end{align}
and $\Theta (t;t) = \begin{pmatrix}
1 & 0 \\
0 &1
\end{pmatrix} $. It follows that the operators $U(t;s), V(t;s)$ satisfy 
\begin{align}
U^*(t;s) U(t;s) - V^*(t;s) V(t;s) =1, \quad U^*(t;s) JV(t;s) = V^*(t;s) JU(t;s) J \,.
\end{align}
and, moreover, 
\begin{align}
\label{eq:prop-Theta}
\Theta (t;s) \mathcal{J}
= \mathcal{J}
\Theta (t;s) , \quad \text{and} \quad   S= \Theta(t;s)^* S \Theta(t;s) \quad \text{with} \quad S = \begin{pmatrix}
1 & 0 \\
0 & -1
\end{pmatrix}
\end{align}
and $\mathcal{J}$ given by \eqref{eq:prop-Theta,J}. Thus $\Theta (t;s)$ is a unitary operator with respect to the scalar product 
\begin{align}
\langle (f_1,f_2), (g_1,g_2) \rangle_{{\rm S}} := \langle (f_1,f_2), S (g_1,g_2) \rangle_{L^2 \oplus L^2}. 
\end{align}
as defined in \eqref{def:scalar-S}. The Bogoliubov dynamics $\Theta (t;s)$ has been widely studied and was proven (see for example \cite{BKS} \footnote{Note that \cite[Theorem 2.2]{BKS} considers a different Bogoliubov dynamics $\widetilde{\Theta}_\infty (t;s)$ w.r.t. unprojected kernels $K_{j,t}$. However, for $\Theta (t;s)$ as defined in \eqref{def:Theta} this statement follows similarly, using here that for the projected kernels $\widetilde{K}_{j,s}$ we immediately get $(\widetilde{K}_{2,s} \varphi_s) = 0$ resp. $(\widetilde{K}_{1,s} J \varphi_s) = 0$. For more details see Section \ref{sec:bogo}.)} to approximate the many-body quantum dynamics in the large particle limit. It satisfies 
\begin{align}
\label{eq:prop-bogo}
\Theta(t;s) ( \varphi_t, \overline{\varphi}_t )  =( \varphi_s, \overline{\varphi}_s ) 
\end{align}
(see \cite[Theorem 2.2]{BKS}). The explicit representation here is implicitly contained in the proofs \cite{BKS,BSS} after replacing the unitary Weyl transform by the excitation map from \cite{LNS,LNSS}; see Section \ref{sec:bogo} and \cite{R} for further details. Well-posedness of \eqref{eq:Theta-s} can be proved, e.g., by combining the calculation presented in Section \ref{sec:bogo} with the well-posedness result for the many-body Bogoliubov dynamics \cite[Theorem 7]{LNS} and we omit the details.

\subsection*{First main result: standard OTOC} We are now ready to state our results. Our first main result is the precise version of \eqref{eq:main} from the introduction. We call a linear operator $A$ on $L^2(\mathbb R^3)$ ``real'', if $AJ=JA$.

\begin{theorem}[OTOC asymptotics]
\label{thm:OTOC}
Let $v$ satisfy \eqref{ass:v} and $\varphi_0 \in H^{4}( \mathbb{R}^3)$ with $\| \varphi_0 \| =1$. Let $t>0$ and $\varphi_t$ denote the solution of the Hartree equation \eqref{def:hartree} with initial data $\varphi_0$. Fix two self-adjoint, real operators $A,B$ satisfying the regularity assumption
\begin{align}
\label{ass:A}
\| (-\Delta + 1) A (-\Delta + 1)^{-1} \| \leq C, \quad \| (-\Delta + 1) B (-\Delta + 1)^{-1} \| \leq C.
\end{align}
Let $\Psi_N$ be given by \eqref{eq:psidefn}, $A_t^{(j)}$ by \eqref{def:Atj} and $\mathcal{B}_0$ by \eqref{def:B0}.\\
 Then we have 
\begin{align}
\label{eq:thm-OTOC}
\lim_{N \rightarrow \infty}  \langle \psi_N, &  (i[A_t^{(1)} , \mathcal{B}_0 ])^2 \psi_N \rangle =\frac{1}{4} \left( \langle (\varphi_0, J \varphi_0), \left[ B, \widetilde{A}_{(t;0)} \right]  (\varphi_0, J \varphi_0) \rangle_{{\rm S}} \right)^2 
\end{align} 
where the scalar product $\langle \cdot, \cdot \rangle_{{\rm S}}$ is given by \eqref{def:scalar-S} and the operator $\widetilde{A}_{(t;0)}$ is 
\begin{align}
\widetilde{A}_{(t;0)} := \Theta (t;0) A \Theta (t;0)^{-1}
\end{align} 
with $\Theta(t;0)$ defined in \eqref{def:Theta}. 
\end{theorem}

On the right-hand side, the effective dynamical operator is thus $\Theta (t;0)$ which has symplectic structure and acts on the smaller effective space $L^2(\mathbb R^3)$.
We remark that the proof is quantitative, i.e., it comes with an error bound for the convergence speed in \eqref{eq:thm-OTOC} which is of the form $N^{-1/2}e^{e^{Ct}}$, see  \eqref{eq:rateref}.

We recall that the practical relevance of formula \eqref{eq:thm-OTOC} is that it expresses the many-body OTOC on the left-hand side by the effective OTOC that only depends on a nonlinear PDE on $\mathbb R^3$ and is therefore calculable numerically. It also opens up the many-body OTOC to the vast mathematical toolkit of dispersive PDE, which leads to an interesting new nonlinear PDE problem that we discuss further in the conclusions section.


The following corollary notes consequences we can obtain from the asymptotic formula \eqref{eq:thm-OTOC}: an expansion for small times and an exponential upper bound on the butterfly velocity for all times.

\begin{corollary}
\label{cor:OTOC}
Under the same assumptions of Theorem \ref{thm:OTOC}, the following holds.

(i) \textbf{Initial rate of information scrambling.} There exists $T>0$ such that for sufficiently small times $\vert t \vert \leq T$ 
\begin{align}
\label{eq:thm-doublecommutator}
\lim_{N \rightarrow \infty}  \langle \psi_N, &  (i[A_t^{(1)} , \mathcal{B}_0 ])^2 \psi_N \rangle  \notag \\
&=  -\frac{1}{2i} \langle \varphi_0, [A,B]  \varphi_0 \rangle ( 1- 2 t \Re \left\langle \varphi_0, B \left[  h_{\varphi_0} + \widetilde{K}_{1,0} - \widetilde{K}_{2,0} J, \;   A  \right] \varphi_0 \right\rangle )+ O(t^2) \; 
\end{align} 
where the operators $\widetilde{K}_{j,t}$ are given by \eqref{def:Kj}.

(ii) \textbf{Bound on butterfly velocity.} We have for all $t \in \mathbb{R}$ 
\begin{align}
\lim_{N \rightarrow \infty}  \langle \psi_N, &  (i[A_t^{(1)} , \mathcal{B}_0 ])^2 \psi_N \rangle  \leq  C e^{C \vert t  \vert }
\end{align} 
where the constant $C$ depends on $ \varphi_0 $ through $\| \varphi_0 \|_{H^1}$.

\end{corollary}

From Corollary \ref{cor:OTOC} (i), we see that if the observables initially commute, $[A,B]=0$, (a common assumption), then the OTOC grows at least quadratically in time initially. If, however, $[A,B]\neq 0$, then the OTOC is initially non-zero and grows ballistically at first. In the second case, we can explicitly describe the speed at which mean-field bosons initially scramble information in terms of the ``energy functional''
\[
\left\langle \varphi_0, B \left[  h_{\varphi_0} + \widetilde{K}_{1,0} - \widetilde{K}_{2,0} J, \;   A  \right] \varphi_0 \right\rangle
\]
that only depends the initial one-body data $\varphi_0\in L^2(\mathbb R^3)$. The point here is that the initial rate of entanglement spreading is explicitly computable from one-body data. Such precise statements are usually not available for many-body systems. (Ballistic information scrambling is expected for generic quantum many-body systems \cite{KH}.) This result also shows that the system is not many-body localized, which would typically correspond to logarithmic in time OTOC-growth \cite{CZHF,FZSZ}. (It would be interesting to study a disordered analog of our model with the goal of finding much slower OTOC-growth; see the Conclusions section.) Linear growth in time at small times is natural for the many-boson system because the subspaces corresponding to different particles are highly connected through the Hamiltonian compared to a system with highly local spatial structure in which the first few terms in a perturbative expansion of the OTOC in $t$ would vanish.

For larger times, part (ii) of Corollary \ref{cor:OTOC} shows that the OTOC is bounded from above exponentially in time. The latter bound is a consequence of Lemma \ref{lemma:f} which is based on a Gronwall argument.
The exponential rate in  Corollary \ref{cor:OTOC}(ii) depends on $\| \varphi_0 \|_{H^1}$. It is an upper bound on the butterfly velocity. For bounded interaction potentials $v\in L^\infty (\mathbb R^3)$, an exponential upper bound on the OTOC, but no exact asymptotics, follows from \cite{ES}. It is an interesting open question under what circumstance the OTOC actually does grow exponentially, because this corresponds to quantum many-body chaos. Theorem \ref{thm:OTOC} turns this into a question about nonlinear dispersive PDE.\\


\subsection*{Additional main results: higher-order OTOCs}
We come to our second main result. While Theorem \ref{thm:OTOC} allows to identify the asymptotics of a special combination of a fourfold product of operators, it is in fact possible to identify the asymptotics of any finite number of operators. This is relevant for the dynamical description of multipartite entanglement \cite{HMY}. The effective description is given by an explicit multivariate Gaussian whose (no real-valued) covariance matrix can be explicitly expressed. Higher-order OTOCs are

This description involves the solution operator $\mathcal{L}_{(t;s)} : L^2( \mathbb{R}^3) \rightarrow L^2( \mathbb{R}^3)$ that satisfies 
\begin{align}
\label{def:Ls}
i \partial_s \mathcal{L}_{(t;s)} =   ( h_{\varphi_s} + \widetilde{K}_{1,s} - \widetilde{K}_{2,s} J)\mathcal{L}_{(t;s)}, 
\end{align}
and in addition 
\begin{align}
\label{def:Lt}
-i \partial_t \mathcal{L}_{(t;s)} =   \mathcal{L}_{(t;s)}( h_{\varphi_t} + \widetilde{K}_{1,t} - \widetilde{K}_{2,t} J), 
\end{align}
with $  \mathcal{L}_{(s,s)}  =1 = \mathcal{L}_{(t;t)}$. With the notation \eqref{def:Theta}, we have $\mathcal{L}_{(t;s)} = ( U(t;s) + J V(t;s))$.
 
\begin{theorem}[Higher OTOCs with regularization]
\label{thm:MCLT}
Let $m \in \mathbb{N}$ and suppose that $A^{(1)},\ldots,A^{m}$ satisfy the same assumptions as $A$ and $B$ in Theorem \ref{thm:OTOC}. Let $\varphi_0,\Psi_N$ be as in Theorem \ref{thm:OTOC}. Given $t_i \in \mathbb{R}$ for $i \in \lbrace 1, \dots, m \rbrace$, we define $\Sigma_{t_1, \dots, t_m} \in  \mathbb{R}^{m \times m}$ by
\begin{align}
\label{def:Sigma}
\left(\Sigma_{t_1, \dots, t_m} \right)_{i,j}= \begin{cases}
\langle \mathcal{L}_{(t_i;0)} q_{t_i} A \varphi_{t_i} , \mathcal{L}_{(t_j;0)} q_{t_j} A \varphi_{t_j} \rangle & \text{if} \quad i\leq j \\
\langle \mathcal{L}_{(t_j;0)} q_{t_j} A\varphi_{t_j}, \mathcal{L}_{(t_i;0)} q_{t_i} A \varphi_{t_i} \rangle& \text{otherwise} 
\end{cases}
\end{align} 
with $\mathcal{L}_{(0;t)}$ given by \eqref{def:Ls}. Assume that $\Sigma_{t_1, \dots, t_m}$ invertible, $g_i \in L^1( \mathbb{R}^3)$, $\widehat{g}_i \in L^1( \mathbb{R}^3, (1+\vert s \vert^5 )ds)$ for all $i \in \lbrace 1, \dots, m \rbrace$ and assume that $A$ is a self-adjoint operator that satisfies \eqref{ass:A}. Let 
\begin{align}
\mathcal{A}_{t} := \sum_{j=1}^N \left( A^{(j)} - \langle \varphi_t, A \varphi_t \rangle \right)
\end{align} 
where $\varphi_t$ denotes the solution to the Hartree dynamics \eqref{def:hartree} with initial data $\varphi_0 = \varphi \in H^4( \mathbb{R}^3)$. Then, 
\begin{align}
\label{eq:thm-MCLT}
\lim_{N \rightarrow \infty} \mathbb{E}  &  \Big[ \prod_{j=1}^m g( \mathcal{A}_{t_j}/\sqrt{N} ) \Big] \notag \\
&=   \frac{1}{2\pi \sqrt{\det \Sigma_{t_1, \dots, t_m}}}\int dx_1 \dots dx_m \; g(x_1) \dots g(x_m)  e^{- (x_1 \dots x_m)^T \Sigma_{t_1, \dots, t_m}^{-1} (x_1 \dots x_m)/2}.  
\end{align} 
\end{theorem}

\begin{remark}
In Theorem \ref{thm:OTOC}, the right-hand side of formula \eqref{eq:thm-OTOC} can also be expressed through the operator $\mathcal L_{(t;0)}$, namely as 
\begin{align}
\label{eq:thm-OTOC2}
\Im \langle q_0 B  \varphi_0, \mathcal{L}_{(t;0)} A \varphi_t \rangle.
\end{align}
This representation is indeed used in the proof of Corollary \ref{cor:OTOC} to derive the r.h.s. of \eqref{eq:thm-doublecommutator} based on the properties \eqref{def:Ls},\eqref{def:Lt}.
\end{remark}

The rate of convergence in Theorem \ref{thm:MCLT} is again doubly exponential in time and $O(N^{-1/2})$ (see \eqref{eq:MCLT-1}). 
Theorem \ref{thm:MCLT} extends the main result of Buchholz-Saffirio-Schlein \cite{BSS} to different times and is proved by similar techniques.

\begin{remark}
    
We compare and contrast the proofs of Theorems \ref{thm:OTOC} and \ref{thm:MCLT} to the mathematical physics literature on Bose gases. Overall, the proofs are based on the idea that the quantum fluctuations around the Hartree evolution can be effectively described by a Bogoliubov dynamics. For higher order correlations, this  was first shown by Buchholz-Saffirio-Schlein \cite{BSS}. We would like to emphasize three main points.

\begin{enumerate}[label=(\roman*)]
    \item A key difference between our setting and that of \cite{BSS} is that our result invovles different times and is formulated in the Heisenberg picture, i.e. we consider the time evolution of operators, whereas \cite{BSS} considers the time evolution of wave functions in the Schr\"odinger picture.  Note that the difference between Schr\"odinger picture and Heisenberg picture is non-trivial when OTOCs are considered.  In this sense, we show that techniques established in the Schr\"odinger picture also extends for different times in the Heisenberg picture for higher order correlations values with respect to suitable initial states \eqref{eq:psidefn}.
    \item 
In contrast to \cite{BSS} we use the so-called excitation map (see Section \ref{sec:flucdyn} for more details) for the description of the quantum fluctuations around the Hartree dynamics. This map has been originally introduced in \cite{LNSS} to derive, similar to \cite{BSS}, an approximation of the fluctuation dynamics on the Fock space by a Fock space Bogoliubov dynamics. We bypass this step, following ideas of \cite{RS}, and show that the fluctuation dynamics effectively acts on the physically relevant quantities as the symplectic dynamics $\Theta(t;0)$ that is defined on $L^2( \mathbb{R}^3) \oplus L^2( \mathbb{R}^3)$. While the approximation argument in \cite{BSS} relies on properties of the Fock-space Bogoliubov dynamics, for us it suffices to analyse properties of the symplectic dynamics $\Theta (t;0)$ (in Section \ref{sec:bogo}).
\item After the mean-field approximation, including the nonlinear Bogoliubov dynamics described in (ii), is implemented, see e.g.\ \eqref{eq:refer}, it is still not obvious that the resulting expression can be rewritten in terms of a commutator at the nonlinear level. Proving this ``only'' involves various algebraic manipulations (displayed, e.g., after \eqref{eq:refer}), but these are conceptually important to firmly link to symplectic dynamics and thereby indeed establish a quantum many-body chaos type of result.
\end{enumerate}
\end{remark}

We may interpret Theorem \ref{thm:OTOC} in the context of non-commutative probability as identifying the finite-dimensional distributions of the non-commutative random process generated by $t\mapsto \mathcal A_t$. We plan to return to this perspective in future work to prove a non-commutative invariance principle (a.k.a.\ functional CLT) for this process, which would expand upon probabilistic results on mean-field bosons \cite{BSS,KRS,R,RS,Rsing}.\\

We now come to our third and last main result. It concerns out-of-time-ordered correlations among $\mathcal{A}_{t_i}$'s, which are essentially symmetrized creation and annihilation operators.

We can think of $\psi_N \in L^2( \mathbb{R}^3)^{\otimes_s N}$ as its emebdding into the bosonic Fock space, and thus asymptotically compute the regularized correlation function through the following theorem. 

\begin{theorem}[Out-of-time-ordered Wick rule]
\label{thm:corrfct}
Under the same assumptions as in Theorem \ref{thm:MCLT} let $m \in \mathbb{N}$, then for all odd $m$ we have 
\begin{align}
\lim_{N \rightarrow \infty} \langle \psi_N, \prod_{i=1}^m \left( \mathcal{A}_{t_i} / \sqrt{N} \right)  \psi_N \rangle = 0,
 \label{eq:odd} 
\end{align}
and for all even $m$
\begin{align}
\lim_{N \rightarrow \infty} \langle \psi_N, \prod_{i=1}^m \left( \mathcal{A}_{t_i} / \sqrt{N} \right)  \psi_N \rangle = \sum_{\pi \in \Pi_m }   (\Sigma_{t_1, \dots, t_m})_{\pi_1,\pi_2}   \cdots (\Sigma_{t_1, \dots, t_m})_{\pi_{m-1},\pi_{m}}, \label{eq:even} 
\end{align}
where the sum runs over the set $\Pi_m :=  \lbrace  \pi \in \mathcal{S}_m \vert \pi_{2i -1} < \pi_{2i} \; \quad \text{for all} \; 1 \leq i \leq m/2 \quad \text{and} \; \pi_{2i-1} < \pi_{2i+1} \quad \text{for all} \; 1 \leq i \leq (m-1)/2 \rbrace $, i.e. over all permutations $\pi = ( \pi_1, \dots \pi_m)$ of the set $\lbrace 1, \dots, m \rbrace$ such that $\pi_{2i -1} < \pi_{2i}$ for all $1 \leq i \leq m/2 $ and $\pi_{2i-1} < \pi_{2i+1}$ for all $1 \leq i \leq (m-1)/2$ and the constants $(\Sigma_{t_1, \dots, t_m})_{i,j}$ are given by \eqref{def:Sigma}.
\end{theorem}

The result is an out-of-time-ordered analog of the standard bosonic Wick rule which expresses expectation values of creation and annihilation operators in quasifree states as the permanent of the one-body correlation matrix (e.g., \cite[Theorem 10.2]{Solovej}). However, unlike in the standard case, here the observables have to suitably symmetrizend but can depend on different times. Theorem \ref{thm:corrfct} then makes precise that this collection of time-evolutions at different times all fit into a single quasi-free approximation and can thus be calculated by the Wick rule through a single time-dependent correlation matrix.

In practical terms, the result says that any $k$-point out-of-time ordered correlation functions of mean-field bosons are explicitly calculable from one-body information depending on the solution to nonlinear Hartree equation. This opens up even multipartite information scrambling to numerical investigations.

At first glance, the quantities in Theorems \ref{thm:OTOC}, \ref{thm:corrfct} appear to be a special case of Theorem \ref{thm:MCLT} when comparing \eqref{eq:thm-OTOC} (and in particular \eqref{eq:thm-OTOC2}) and \eqref{eq:even} with \eqref{eq:thm-MCLT}. However, the restrictive regularity assumptions on the test function $g$ in  Theorem \eqref{thm:MCLT} prevent us from deriving Theorem \eqref{thm:OTOC} (which corresponds to the case where $g$ is the identity map $x \mapsto x$). Accordingly, we have to prove both theorems separately. The details of the proofs are given in Section \ref{sec:MCLT} resp. Section \ref{sec:OTOC}, \ref{sec:corr}. 
The proof of Theorem \eqref{thm:corrfct} uses the same techniques as those of Theorems \ref{thm:OTOC} and \ref{thm:MCLT}, resulting in the same rate of convergence (i.e., exponential in time and $O(N^{-1/2})$ in the total number of particles). 

\subsection*{Structure of the paper} As mentioned earlier, the proofs of Theorems \ref{thm:OTOC}-\ref{thm:corrfct} are based on similar ideas, which we explain first on the proof of Theorem \ref{thm:MCLT} in Section and then apply to the remaining two theorems. 

Before proving the main theorems, we first collect useful properties of the operators $ \mathcal{L}_{(t;s)}$ in Section \ref{sec:L}. For this we need properties of the kernels $\widetilde{K}_{j,s}$ defined in \eqref{def:Kj} which we prove here and in some cases take from \cite{RS}. The following Section \ref{sec:flucdyn} is dedicated to the quantum fluctuations around the condensate that we need to study to then prove Theorem \ref{thm:MCLT} in Section \ref{sec:MCLT}, Theorem \ref{thm:corrfct} in Section \ref{sec:corr}, Theorem \ref{thm:OTOC} in Section \ref{sec:OTOC} and finally Corollary \ref{cor:OTOC} in Section \ref{sec:cor}.

\section{Properties of the nonlinear Bogoliubov dynamics $ \mathcal{L}_{(t;s)}$}
\label{sec:L}

In this section we collect and prove properties of the operators $ \mathcal{L}_{(t;s)}$ for which we first collect and prove properties of the kernels $\widetilde{K}_{j,t}$ defined in \eqref{def:Kj}. The properties mainly rely on propagation of regularity of the solution of the Hartree equation. By conservation of energy and the assumption on the potential $H^1$-norm of the solution of the Hartree equation \eqref{def:hartree} is bounded uniformly in time, i.e. there exists  $C>0$ such that 
\begin{align}
\| \varphi_s \|_{H^1} \leq C \| \varphi_0 \|_{H^1} \; . 
\end{align}
We will furthermore need propagation of higher Sobolev norms, that are not uniformly in time. However it is well known that 
\begin{align}
\label{eq:Sobolev-k}
\| \varphi_s \|_{H^k} \leq C  e^{C s}\| \varphi_0 \|_{H^k}
\end{align}
for all $k \geq 2$ (see for example \cite{Caz}). These estimates together with the assumption on the potential result in (see for example \cite[Lemma 3.1]{RS})
\begin{align}
\label{eq:estimates-pot}
 \| v * \vert \varphi_t \vert^2 \|_\infty,   \| v^2 * \vert \varphi_t \vert^2 \|_\infty\leq C_1 
\end{align}
and 
\begin{align}
\label{eq:estimate-K}
\| K_{j,t} \|_{L^2 ( \mathbb{R}^3 \times \mathbb{R}^3)} \leq C _2,  \quad \|  K_{j,t} f \|_{H^2} \leq C_3 \; e^{C_3 \vert t \vert}\| f \|_{H^2} \; . 
\end{align}
From the proof in \cite[Lemma 3.1]{RS} it turns out that the constant $C_1$ (resp. $C_2,C_3$) depends on $\varphi_0$ through $\| \varphi_0 \|_{H^3}$ (resp. $\| \varphi_0 \|_{H^1}$ and $\| \varphi_0 \|_{H^4}$). 
We will furthermore need the following estimates on the time derivatives of $h_{\varphi_t}, K_{j,t}, \varphi_t$.

\begin{lemma}
\label{lemma:K} 
For $t \in \mathbb{R}$ let $v$ satisfy \eqref{ass:v} and  $\varphi_t$ denote the solution to the Hartree equation with $\varphi_0 \in H^2 ( \mathbb{R}^3)$. Then there exists $C>0$ such that for any $g \in H^2( \mathbb{R}^3)$
\begin{align}
\|  h_{\varphi_t} g \| \leq C \| g  \|_{H^2}, \; \| \dot{h}_{\varphi_t} g \|_2 \leq C  e^{C t }\| g \|_{H^2}
\end{align}
and for $K_{j,t}$ defined in \eqref{def:Kj} we have 
\begin{align}
\label{eq:estimates-Kj-dot}
\| \dot{K}_{j,t} \|_{L^2 ( \mathbb{R}^3 \times \mathbb{R}^3)} \leq C e^{C t}  \; . 
\end{align}
\end{lemma}

\begin{proof} To prove the first estimate we observe that by \eqref{def:hartree} 
\begin{align}
\| h_{\varphi_t} g \|_2 \leq \| g \|_{H^2} + \| v * \vert \varphi_t \vert^2 \|_\infty \| g \|_2 \leq C \| g \|_{H^{2}} 
\end{align}
where we used \eqref{eq:estimates-pot}. For the remaining estimates note that by the above estimate and \eqref{eq:Sobolev-k}
\begin{align}
\label{eq:estimates-dot-phi}
\| \dot{\varphi}_t \|_{2} \leq \| h_{\varphi_t}\varphi_t \|_2 \leq C \| \varphi_t \|_{H^2} \leq  C e^{C t} \| \varphi_0 \|_{H^2} \; . 
\end{align}
Then we get 
\begin{align}
\| \dot{h}_{\varphi_t} g \|_2 = \| (v * ( 2 \Re (\dot{\varphi}_t \varphi_t )) g \|_2 \leq \| v * ( 2 \Re (\dot{\varphi}_t \varphi_t ) \|_\infty \| g \|_2 
\end{align}
and by assumption on $v$
\begin{align}
\vert v*  2 \Re (\dot{\varphi}_t \varphi_t ) \vert (x) &\leq \int \vert v(x-y) \vert \vert \dot{\varphi}_t (y)\vert  \; \vert \varphi_t (y) \vert dy \notag \\
&\leq \int \vert v(x-y) \vert^2 \vert \varphi_t (y) \vert^2 + C \| \dot{\varphi}_t \|_2^2. 
\end{align}
Then from \eqref{eq:estimates-dot-phi} 
\begin{align}
\vert v*  2 \Re (\dot{\varphi}_t \varphi_t ) \vert (x) &\leq C \| \varphi_t \|_{H^2}
\end{align}
which by \eqref{eq:Sobolev-k} yields
\begin{align}
\| \dot{h}_{\varphi_t} g \|_2 \leq e^{C t } \| g \|_2 \; . 
\end{align}
To prove \eqref{eq:estimates-Kj-dot}, we observe that 
\begin{align}
\| \dot{K}_{j,t} \|_{L^2 ( \mathbb{R}^3 \times \mathbb{R}^3)}^2 \leq C \int dxdy \vert \dot{\varphi}_t (x) \vert^2 \vert v(x-y) \vert^2 \vert \varphi_t (y) \vert^2  = \langle \dot{\varphi}_t, (v^2 * \vert \varphi_t \vert^2 ) \dot{\varphi}_t  \rangle 
\end{align}
and we conclude by \eqref{eq:estimates-pot} and \eqref{eq:estimates-dot-phi} with 
\begin{align}
\| \dot{K}_{j,t} \|_{L^2 ( \mathbb{R}^3 \times \mathbb{R}^3)}^2  \leq C \|v^2* \vert \varphi_t \vert^2 \|_\infty \| \dot{\varphi}_t\|_2 \leq C e^{C t} \| \varphi_0 \|_{H^2} \; . 
\end{align}
\end{proof}

We remark that \eqref{eq:estimates-Kj-dot} easily generalizes for the projected kernels $\widetilde{K}_{j,t}$, too. 

 Next we recall some properties of the operator $\mathcal{L}_{(t;s)}$ defined in \eqref{def:Ls} proven in \cite[Lemma 3.2]{RS}. It should be noted that \cite[Lemma 3.2]{RS} is formulated for a specific function with given initial conditions. The proof given there easily generalizes to the Lemma as stated here.  

\begin{lemma}[Lemma 3.2 \cite{RS}] \label{lemma:f}
Let $v$ satisfy \eqref{ass:v} and $\varphi_0 \in H^{4} ( \mathbb{R}^3)$. For $s,t\geq 0$ let $\mathcal{L}_{(t;s)}$ be given by \eqref{def:Ls}. Then there exists $C>0$ such that for $g \in H^2 ( \mathbb{R}^3)$ we have 
\begin{align}
\label{eq:estimate-f}
\| \mathcal{L}_{(t;s)} g \|_2 \leq C_1 e^{C_1 \vert t -s \vert} \| g \|_2, \quad \| \mathcal{L}_{(t;s)}g \|_{H^k} \leq C_2 e^{e^{C_2 \max  \lbrace t,s\rbrace }} \| g\|_{H^2 }  
\end{align}
where $C_1$ (resp. $C_2$) depends on $\varphi_0$ through $\| \varphi_0 \|_{H^1}$ (resp. $\| \varphi_0 \|_{H^4}$). 
\end{lemma}

\section{Many-body time evolution}
\label{sec:flucdyn}

We are interested in the dynamics of the fluctuations around the condensate. For this we embed the system in the bosonic Fock space $\mathcal{F}$ that is equipped with standard creation and annihilation operators $a^*(f), a(g)$ satisfying for any $f,g \in L^2( \mathbb{R}^3)$ the standard commutation relations 
\begin{align}
\label{eq:comm}
[a(g), a^*(f)] = \langle g,f \rangle, \quad [a(f),a(g)] = 0 = [a^*(f), a^*(g)] \; .
\end{align}
To describe the fluctuations around the condensate, we factor out the condensates contributions through the unitary $\mathcal{U}_{N,t} : L_s^2( \mathbb{R}^{3N}) \rightarrow \mathcal{F}_{\perp \varphi_t}^{\leq N}$. Any $\Psi_N \in L^2( \mathbb{R}^{3N})$ that can be uniquely decomposed as 
\begin{align}
\Psi_N = \eta_0 \varphi^{\otimes N}_t + \dots + \eta_N \quad \text{with} \quad \eta_j \in L_{\perp \varphi_t}^{2} ( \mathbb{R}^{3})^{\otimes_s j}
\end{align}
and one sets
\begin{align}
\label{def:U}
\mathcal{U}_{N,t} \Psi_N = \lbrace \eta_0, \dots, \eta_N \rbrace
\end{align}
The unitary has the following properties for $f,g \in L_{\perp \varphi_t}^{2} ( \mathbb{R}^{3})$ (see \cite{LNS})
\begin{align}
\mathcal{U}_{N,t} a^*( \varphi_t) a( \varphi_t) \mathcal{U}_{N,t}^* &= N  - \mathcal{N} \notag \\
\mathcal{U}_{N,t} a^*( f) a(g) \mathcal{U}_{N,t}^* &=  a^*(f) a(g)  \label{eq:prop-U1}
\end{align}
where $\mathcal N$ is the number of particles operator on the truncated Fock space $\mathcal{F}_{\perp \varphi_t}^{\leq N}$ (which equals the number of excitations). Moreover, 
\begin{align}
\mathcal{U}_{N,t} a^*( \varphi_t) a(f) \mathcal{U}_{N,t}^* &= \sqrt{N  - \mathcal{N} } a(f) = \sqrt{N} b(f)  \notag \\
\mathcal{U}_{N,t} a^*( f) a(\varphi_t) \mathcal{U}_{N,t}^* &=  a^*(f)\sqrt{N  - \mathcal{N} } = \sqrt{N} b^*(f) \label{eq:prop-U2}
\end{align}
where we introduced the modified creation and annihilation operators $b^*(f), b(g)$ that satisfy for $f,g \in L^2( \mathbb{R}^3)$ the modified creation and annihilation commutation relations 
\begin{align}
\label{eq:comm-b}
[b(f), b(g) ] = [b^*(f), b^*(g)] = 0, \quad [b(f), b^*(g) ] = \langle f,g\rangle \Big( 1 - \frac{\cN }{N}\Big) - \frac{1}{N} a^*(g)a(f) \; . 
\end{align}
On the one hand, the correction of the modified commutation relations compared to the standard ones \eqref{eq:comm} leads to technical difficulties in handling the modified creation and annihilation operators $b^*(f),b(g)$. On the other hand, working with $b^*(f), b(g)$ allows to work in the truncated Fock space $\mathcal{F}_{\perp \varphi_t}^{\leq N}$ where the number of particles operator $\cN$ is naturally bounded by $N$ which will be crucial for our analysis and, moreover, where the number of excitations  equals the number of particles $\mathcal{N}$. 
 
\subsection{Fluctuation dynamics} With the unitary \eqref{def:U}, we now define for $t \in \mathbb{R}$ the fluctuation dynamics (i.e. the dynamics describing the fluctuations around the condensate of the gas) 
\begin{align}
\label{def:flucdyn}
\mathcal{W}_N (t;0)  = \mathcal{U}_{N,t} e^{-iH_Nt} \mathcal{U}_{N,0}^* \; . 
\end{align}
and it follows from \eqref{eq:prop-U1},\eqref{eq:prop-U2} that 
\begin{align}
i \partial_t \mathcal{W}_N(t;0) = \mathcal{G}_N (t) \mathcal{W}_N(t;0) 
\end{align}
where 
\begin{align}
\label{def:G}
\mathcal{G}_N(t) = \mathbb{H}(t) + \sum_{j=1}^3 R_{N,t}^{(j)}
\end{align}
and 
\begin{align}
\label{def:H}
\mathbb{H} (t) = d \Gamma (h_{\varphi_t} + \widetilde{K}_{1,t} ) + \frac{1}{2}\int dxdy \left( \widetilde{K}_{2,t}(x;y) b_x^*b_y^* + \overline{\widetilde{K}}_{2,t} (x;y) b_xb_y \right)
\end{align}
with $h_{\varphi_t} = - \Delta + (v* \vert \varphi_t \vert^2)$ and $\widetilde{K}_{j,t} $ defined in  \eqref{def:tilde-Kj} and \eqref{def:Kj}. 
The remainders are given by
\begin{align}
\label{def:R1}
R_{N,t}^{(1)} =& \frac{1}{2} d \Gamma (q_t \left[ v * \vert \varphi_t \vert^2 + \widetilde{K}_{1,t} - \mu_t \right] q_t ) \frac{1-\cN (t)}{N} \notag \\
&+ \frac{\cN(t)}{\sqrt{N}}  b( q_t (v*\vert \varphi_t \vert^2) \varphi_t) \notag \\
&+{\rm h.c.} 
\end{align}
where $2 \mu_t = \int v(x-y) \vert \varphi_t (x)\vert^2 \vert \varphi_t (y)\vert^2$ and 
\begin{align}
\label{def:R23}
R_{N,t}^{(2)} =& \frac{1}{\sqrt{N}}\int v(x-y)\varphi_t (y) a^*(q_{t,x}) a (q_{t,x}) b(q_{t,y}) dxdy + {\rm h.c.} \notag \\
R_{N,t}^{(3)} =& \frac{1}{N} \int v(x-y)\ a^*(q_{t,x}) a^*(q_{t,y}) a (q_{t,x}) a(q_{t,y})  dxdy \;. 
\end{align}
The fluctuation dynamics approximately preserves the number of particles as the following Lemma shows.

 \begin{lemma}
 \label{lemma:WN}
Let $v \leq C (1-\Delta)$ and $\varphi_t$ denote the solution to the Hartree equation \eqref{def:hartree} with initial data $\varphi_0 \in H^2( \mathbb{R}^3)$. Let $k \in \mathbb{N}$, then there exists $C_k>0$ such that 
\begin{align}
\label{eq:WN-Nk}
\mathcal{W}_N^* (t;0) ( \mathcal{N} + 1 )^k \mathcal{W}_N (t;0)  \leq C_k e^{C \vert t \vert} ( \mathcal{N} + 1)^k \; . 
\end{align}
as quadratic form on the truncated Fock space $\mathcal{F}_{\perp \varphi_t}^{\leq N}$. 
\end{lemma}

\begin{proof} For $k=0$, the proof can be found in \cite[Lemma 14]{LNS}. The proof is based on Duhamel's formula 
\begin{align}
\mathcal{W}_N^* (t;0) \mathcal{N}  \mathcal{W}_N (t;0)  - \mathcal{N} = \int_0^t ds \; \mathcal{W}_N^* (s;0) \left[ \mathcal{G}_N (s), \mathcal{N} \right]  \mathcal{W}_N (s;0) 
\end{align}
and the bound (see \cite[Lemmas 9,13]{LNS})
\begin{align}
\label{eq:sc}
i \Big[ \cN , \mathcal{G}_N (s) \Big] \leq  C  (\cN  +1) \; . 
\end{align}
We will show that this holds true for the second nested commutator, too, i.e. that for all $j \in \mathbb{Z}$ there exists $C_j>0$ such that 
\begin{align}
\label{eq:dc}
\big\vert \langle \xi, \; ( \cN  + 1)^{j/2}  \Big[ \cN , \mathcal{G}_N (s) \Big]\psi \rangle \big\vert \leq  C_j \| \xi \| \; \| (\cN  +1)^{(j+2)/2} \psi \| \; .  
\end{align}
Then \eqref{eq:WN-Nk} follows for general $k\geq 0$ as 
\begin{align}
\mathcal{W}_N^* (t;0) (\mathcal{N}  + 1)^k \mathcal{W}_N (t;0)  - (\mathcal{N} + 1)^k = \int_0^t ds \; \mathcal{W}_N^* (s;0) \left[ \mathcal{G}_N (s), (\mathcal{N}  + 1)^k \right]  \mathcal{W}_N (s;0) 
\end{align}
and thus with the properties of the commutator 
\begin{align}
\mathcal{W}_N^* (t;0)&  (\mathcal{N}  + 1)^k \mathcal{W}_N^* (t;0)  - (\mathcal{N} + 1)^k \notag \\
=& \int_0^t ds \; \mathcal{W}_N^* (s;0) (\mathcal{N}  + 1)^{k-1} \left[ \mathcal{G}_N (s), \mathcal{N}  \right]  \mathcal{W}_N (s;0) \notag \\
&+\int_0^t ds \; \mathcal{W}_N^* (s;0)  \left[ \mathcal{G}_N (s), \mathcal{N}  \right] (\mathcal{N}  + 1)^{k-1} \mathcal{W}_N (s;0) \; 
\end{align}
and \eqref{eq:WN-Nk} follows with \eqref{eq:dc} from Gronwall's inequality. 

It remains to prove \eqref{eq:dc}. By definition \eqref{def:H}, \eqref{def:R1}, \eqref{def:R23} we find 
\begin{align}
\big[ \mathcal{G}_N (t), \mathcal{N} \big] =& \int dxdy \left( \widetilde{K}_{2,t}(x;y) b_x^*b_y^* - \overline{\widetilde{K}}_{2,t} (x;y) b_xb_y \right) \notag \\
&+ \frac{\cN }{\sqrt{N}}  b( q_t (v*\vert \varphi_t \vert^2) \varphi_t) - {\rm h.c.}  \notag \\
&+ \frac{1}{\sqrt{N}}\int v(x-y)\varphi_t (y) a^*(q_{t,x}) a (q_{t,x}) b(q_{t,y}) dxdy - {\rm h.c.} \label{eq:L-N_comm}
\end{align}
Since $\cN  b^*(f) = b^*(f) ( \cN  + 1)$ we find for the first term of the r.h.s. 
\begin{align}
\langle \xi, \; ( \cN  + 1)^{j/2} \int dxdy \widetilde{K}_{2,t}(x;y) b_x^*b_y^* \psi \rangle  =\langle \xi, \;  \int dxdy \widetilde{K}_{2,t}(x;y) b_x^*b_y^* ( \cN  +3 )^{j/2}\psi \rangle
\end{align}
and 
\begin{align}
\vert \langle \xi, \; ( \cN  + 1)^{j/2} \int dxdy \widetilde{K}_{2,t}(x;y) b_x^*b_y^* \psi \rangle \vert   &\leq  \| K_{2,t}\|_2 \| \xi \| \; \|( \cN +1)^{(j+2)/2} \psi \| \notag \\
 &\leq C \| \xi \| \; \|( \cN +1)^{(j+2)/2} \psi \| 
\end{align}
where we used that by \cite[Lemma 10]{LNS}
\begin{align}
 \| K_{2,t}\|_2 \leq C \; .
\end{align}
The hermitian conjugate follows similarly. For the second term of the r.h.s. of \eqref{eq:L-N_comm} we find similarly using that $\psi, \xi \in \mathcal{F}_{\perp \varphi_t}^{\leq N}$
\begin{align}
\vert\langle  &\xi, ( \cN  + 1)^{j/2} \frac{\cN }{\sqrt{N}}  b( q_t (v*\vert \varphi_t \vert^2) \varphi_t) \psi \rangle \vert \notag \\
& \leq \|  q_t (v*\vert \varphi_t \vert^2) \varphi_t \|_2 \| \xi \| \; \| ( \cN  + 1)^{(j+1)/2} \psi \| \leq C \| \xi \| \; \|  ( \cN  + 1)^{(j+1)/2} \psi \|
\end{align}
where we used that 
\begin{align}
\| q_t  (v*\vert \varphi_t \vert^2) \varphi_t \|_2 \leq \|  (v*\vert \varphi_t \vert^2) \|_\infty \| \varphi_t \|_2 \leq C
\end{align}
by \cite[Lemma 10]{LNS}. The hermitian conjugate follows again similarly. The remaining term of \eqref{eq:L-N_comm} can be estimated by
\begin{align}
\vert \langle \xi,  & ( \cN  + 1)^{j/2}  \frac{1}{\sqrt{N}}\int v(x-y)\varphi_t (y) a^*(q_{t,x}) a (q_{t,x}) b(q_{t,y}) dxdy \; \psi \rangle \vert \notag \\
\leq&  \left( N^{-1} \int (v * \vert \varphi_t\vert^2 ) (x) \| a (q_{t,x} ) \xi \|^2 \right)^{1/2} \left( \int dxdy  \| a(q_{t,x}) b(q_{t,y})  ( \cN  + 1)^{j/2} \psi \|^2 \right)^{1/2} \notag \\
\leq& C \| \xi \| \| ( \mathcal{N} + 1)^{(j+2)/2}  \psi \| 
\end{align}
and thus we arrive at \eqref{eq:dc}. 
\end{proof}

\subsection{Connection between many-body and nonlinear Bogoliubov dynamics} \label{sec:bogo} For Lemma \ref{lemma:WN} it is important to work on the truncated Fock space. However, in the limit of a large number of particles $N \rightarrow \infty$, the fluctuation dynamics is asymptotically well described by $\mathcal{W}_\infty (t;0)$ that is a dynamics quadratic in creation and annihilation operators and given by 
\begin{align}
\label{def:Winfty}
i \partial_t \mathcal{W}_\infty (t;0) = \widetilde{\mathbb{H}} \mathcal{W}_\infty (t;0)
 \end{align}
 where 
\begin{align}
\label{def:Htilde}
\widetilde{\mathbb{H}} = d \Gamma (h_{\varphi_t} + \widetilde{K}_{1,t} ) + \frac{1}{2}\int dxdy \left( \tilde K_{2,t}(x;y) a_x^*a_y^* + \overline{\tilde K}_{2,t} (x;y) a_xa_y \right)
\end{align}
with $h_{\varphi_t}$ given by \eqref{def:hartree} and $\tilde K_{j,s}$ defined in \eqref{def:tilde-Kj} and \eqref{def:Kj}. Note that the generator $\widetilde{\mathbb{H}}$ is quadratic in standard creation and annihilation operators. Dynamics with generators that are quadratic in standard creation and annihilation operators are known to give rise to a Bogoliubov dynamics with explicit action on creation and annihilation operators.

Here we give a brief overview about the properties of $\mathcal{W}_\infty (t;s)$ and its relation to the Bogoliubov dynamics \eqref{def:Theta} that is widely studied in the literature. We follow the approach of \cite{BKS} where this perspective seems to originate. Namely, it follows by arguments similar\footnote{ In \cite{BKS} a slightly different Bogoliubov dynamics $\widetilde{\mathcal{W}}_\infty (t;s)$ was formulated w.r.t. to the unprojected kernels $K_{j,s}$ (see \eqref{def:Kj} ) instead of the projected kernels $\widetilde{K}_{j,s}$ of $\mathcal{W}_\infty (t;s)$ defined in \eqref{def:Winfty}. However the arguments in \cite{BKS} can be generalized to $\mathcal{W}_\infty (t;s)$. This is carried out in detail in \cite[Proposition 1.3]{R} for a more general class of dynamics with quadratic generators that $\mathcal{W}_\infty (t;s)$ belongs to by Lemma \ref{lemma:K}. See also \cite{LNS}.} to those in \cite{BKS} that $\Theta (t;s)$ is a  Bogoliubov transform associated to the asymptotic (quadratic) dynamics $\mathcal{W}_\infty (t;s)$. More precisely $\Theta (t;s) : L^2( \mathbb{R}^3) \oplus L^2( \mathbb{R}^3) \rightarrow L^2( \mathbb{R}^3) \oplus L^2( \mathbb{R}^3)$ is a bounded linear map satisfying for all $f,g \in L^2( \mathbb{R}^3)$ and $A(f,g) = a(f) + a^*(\overline{g})$ 
\begin{align}
\mathcal{W}_\infty^*(t;s) A(f,g) \mathcal{W}_\infty (t;s) = A( \Theta(t;s) (f,g) ), \; 
\end{align}
where $\mathcal{W}_\infty (t;s)$ is defined in \eqref{def:Winfty}. Moreover it follows from arguments presented in \cite{BKS} that the Bogoliubov map $\Theta (t;s)$ satisfies \eqref{eq:Theta-s} for all $s,t \in \mathbb{R}$
\begin{align}
\label{eq:Theta-t}
i\partial_t \Theta (t;s) =- \Theta (t;s)  \mathcal{T}_t  \quad \text{with} \quad  \mathcal{T}_t \quad \text{given by} \quad \eqref{eq:Theta-s}
\end{align}
and $\Theta (s;s) =\begin{pmatrix}
 1 & 0 \\
 0 &1 
\end{pmatrix} $.\footnote{Note that there is a sign discrepancy between \eqref{eq:Theta-t} and \cite[formula (2.12)]{BKS}.}  One obtains the existence of $\Theta (t,s)$ such that \eqref{eq:prop-Theta} is satisfied and then \eqref{eq:Theta-t} follows from the identity 
\begin{align}
A( i \partial_t \Theta(t,s) (f,g)) =& i \partial_t \mathcal{W}_\infty (t;s)^* A(f,g) \mathcal{W}_\infty (t;s) = -\mathcal{W}_\infty (t;s) [ \widetilde{\mathbb{H}}, A(f,g) ] \mathcal{W}_\infty (t;s)  
\end{align}
where $\mathcal{W}_\infty (t,s), \widetilde{\mathbb{H}}$ are given by \eqref{def:Winfty}, \eqref{def:Htilde}. With similar computations as in \eqref{eq:commH-phi}, the commutator can be explicitly computed 
\begin{align}
A( i \partial_t \Theta(t,s) (f,g))  =& \mathcal{W}_\infty^* (t;s)A(\mathcal{T}_t (f,g) ) \mathcal{W}_\infty (t;s) =  A(\Theta(t;s) \mathcal{T}_t (f,g)  ) 
\end{align}
and \eqref{eq:Theta-t} follows. The second identity \eqref{eq:Theta-s} follows similarly with 
\begin{align}
0 =i \partial_s \mathcal{W}_\infty (t;s) A( \Theta(t,s) (f,g))\mathcal{W}_\infty (t;s)^* = i \partial_s \mathcal{W}_\infty (s;t)^* A( \Theta(t,s) (f,g))\mathcal{W}_\infty (s;t) 
\end{align}
and thus similarly as before 
\begin{align}
0 = \mathcal{W}_\infty (s;t)^* \left( -[ \widetilde{\mathbb{H}}, A( \Theta(t,s) (f,g))] +  A( i\partial_s\Theta(t,s) (f,g)) \right)  \mathcal{W}_\infty (s;t)  
\end{align}
and we conclude by \eqref{eq:Theta-s}.
In particular, it follows that 
\begin{align}
\Theta(t;s) ( \varphi_t, \overline{\varphi}_t )  =( \varphi_s, \overline{\varphi}_s ) 
\end{align}
as in \cite[Theorem 2.2]{BKS}.\footnote{Note that  - as before - \cite[Theorem 2.2]{BKS} considers a different Bogoliubov dynamics $\widetilde{\Theta}_\infty (t;s)$ w.r.t. unprojected kernels $K_{j,t}$. However, for $\Theta (t;s)$ as defined in \eqref{def:Theta} this statement follows similarly, using here that for the projected kernels $\widetilde{K}_{j,s}$ we immediately get $(\widetilde{K}_{2,s} \varphi_s) = 0$ resp. $(\widetilde{K}_{1,s} J \varphi_s) = 0$.} Moreover, the second identity of \eqref{def:Ls} follows.

\section{Multivariate Gaussian: Proof of Theorem \ref{thm:MCLT} }
\label{sec:MCLT}

In this section we prove Theorem \ref{thm:MCLT}. The ideas and steps used for this proof will also be important later for the proofs of Theorems \ref{thm:OTOC} and Theorem \ref{thm:corrfct} in Sections \ref{sec:OTOC} resp. \ref{sec:corr}.

\begin{proof}[Proof of Theorem \ref{thm:MCLT} ]
 Theorem \ref{thm:MCLT} follows with the observation that for any $g_i \in L^1( \mathbb{R}^3)$ 
\begin{align}
\mathbb{E}  \Big[ \prod_{j=1}^m g_i( \mathcal{A}_{t_j}/\sqrt{N} ) \Big] 
=& \int_{\mathbb{R}^m} \widehat{g}_1(\lambda_1)  \cdots \widehat{g}_m( \lambda_m) \; \langle \psi_N,  e^{i \frac{\lambda_1}{\sqrt{N}} A_{t_1}} \cdots e^{i \frac{\lambda_m}{\sqrt{N}} A_{t_m}} \psi_N \rangle \;  d \lambda_1 \dots d \lambda_m \notag \\
=& \int_{\mathbb{R}^m} \widehat{g}_1(\lambda_1)  \cdots \widehat{g}_m( \lambda_m) \;  \mathbb{E}  \Big[  \prod_{j=1}^m e^{i \frac{\lambda_j}{\sqrt{N}} A_{t_j}}  \Big]d \lambda_1 \dots d \lambda_m 
\end{align}
by the bound on the characteristic function
\begin{align}
\label{eq:MCLT-1}
\Big\vert \mathbb{E} & \Big[  \prod_{j=1}^m e^{i \frac{\lambda_j}{\sqrt{N}} A_{t_j}}  \Big]  - e^{- \lambda^T \Sigma_{t_1, \dots, t_m}\lambda /2} \Big\vert \leq  C_m e^{e^{C \vert t \vert}} N^{-1/2} \prod_{j=1}^{m} ( 1+ \lambda_j^2 )^{5/2} 
\end{align}
where $\lambda = ( \lambda_i, \cdots, \lambda_m) \in \mathbb{R}^m$ and $\Sigma_{t_1, \dots, t_m} \in  \mathbb{R}^{m \times m}$ is given by \eqref{def:Sigma}. 

Accordingly, it is sufficient to show formula \eqref{eq:MCLT-1}, which we will now proof in the following. For this we write the expectation value in the bosonic Fock space using the notation 
\begin{align}
A_{t} =
e^{itH_N} (A - \langle \varphi_t, A \varphi_t \rangle ) e^{-itH_N}
\end{align} 
as 
\begin{align}
\mathbb{E} & \Big[  \prod_{j=1}^m e^{i \frac{\lambda_j}{\sqrt{N}} A_{t_j}}  \Big]   = \langle \psi_N, \prod_{j=1}^m e^{i \frac{\lambda_j}{\sqrt{N}} d \Gamma (A_{t_j})}  \psi_N \rangle \; . 
\end{align}
Using the unitary map $\mathcal{U}_{N,0}$ defined in \eqref{def:U} we have $\psi_N = \mathcal{U}_{N,0}^* \Omega $, i.e. 
\begin{align}
\mathbb{E}  \left[ \prod_{j=1}^m e^{i \frac{\lambda_j}{\sqrt{N}} A_{t_j}} \right]  =& \langle \mathcal{U}_{N,0}^* \Omega, \prod_{j=1}^m e^{i \frac{\lambda_j}{\sqrt{N}} d \Gamma (A_{t_j})}   \mathcal{U}_{N,0}^* \Omega \rangle =\langle \ \Omega, \prod_{j=1}^m e^{i\frac{\lambda_j}{\sqrt{N}} \mathcal{U}_{N,0} d \Gamma (A_{t_j}) \mathcal{U}_{N,0}^*}  \Omega \rangle
\end{align}
With the definition of the fluctuation dynamics \eqref{def:flucdyn} we find that 
\begin{align}
\frac{1}{\sqrt{N}}  & \mathcal{U}_{N,0} d \Gamma (A_{t_j}) \mathcal{U}_{N,0}^* \notag \\
&= \frac{1}{\sqrt{N}} \mathcal{W}_N (t_j;0)^* \mathcal{U}_{N,t_j} d\Gamma \left( A - \langle \varphi_{t_j}, A \varphi_{t_j} \rangle  \right)\mathcal{U}_{N,t_j}^* \mathcal{W}_N (t_j;0) \; . 
\end{align}
The properties of the unitary $\mathcal{U}_{N,t_j}$ (see \eqref{eq:prop-U1},\eqref{eq:prop-U2}) show that 
\begin{align}
\mathcal{U}_{N,t_j} & d\Gamma \left( A - \langle \varphi_{t_j}, A \varphi_{t_j} \rangle  \right)\mathcal{U}_{N,t_j}^*\notag \\
 =& d \Gamma (q_{t_j} \left( A - \langle \varphi_{t_j}, A \varphi_{t_j} \rangle  \right) q_{t_j} ) + b( q_{t_j} A \varphi_{t_j}) + b^*( q_{t_j} A \varphi_{t_j})  \notag \\
=& \sqrt{N} \phi_+ ( q_{t_j} A \varphi_{t_j} ) + \widetilde{A}_{t_j} \label{eq:B}
\end{align}
where we introduced the notation $\phi_+ (h ) = b^*(h) + b(h)$ for any $h \in L^2( \mathbb{R}^3)$ and 
\begin{align}
\widetilde{A}_{t_j} =& d \Gamma (q_{t_j} \left( A - \langle \varphi_{t_j}, A \varphi_{t_j} \rangle  \right) q_{t_j} ) \; . 
\end{align}
Thus we get 
\begin{align}
\mathbb{E}  & \Big[ \prod_{j=1}^m e^{i \frac{\lambda_j}{\sqrt{N}} A_{t_j}} \Big]  =\langle  \Omega, \prod_{j=1}^m e^{ i\lambda_j\mathcal{W}_N (t_j;0)^* ( \phi_+ ( q_{t_j} A \varphi_{t_j} ) + N^{-1/2} \widetilde{A}_{t_j})\mathcal{W}_N (t_j;0)  } \Omega \rangle. 
\end{align}

\subsection*{Step 1} In the first step we show that the operator $\widetilde{A}_{t_j}$ is negliable in the limit $N \rightarrow \infty$, i.e. that 
\begin{align}
 \lambda_j\mathcal{W}_N (t_j;0)^* ( \phi_+ ( q_{t_j} A \varphi_{t_j} ) + N^{-1/2} \widetilde{A}_{t_j})\mathcal{W}_N (t_j;0)  \approx  \lambda_j\mathcal{W}_N (t_j;0)^* \phi_+ ( q_{t_j} A \varphi_{t_j} ) \mathcal{W}_N (t_j;0). 
\end{align}

We introduce the notation $h_{t_k} =  q_{t_k} A \varphi_{t_k}$ and compare the expectation value above with 
\begin{align}
\mathbb{E}  & \Big[ \prod_{j=1}^m e^{i \frac{\lambda_j}{\sqrt{N}} A_{t_j}} \Big]  - \langle  \Omega, \prod_{j=1}^m e^{i \lambda_j\mathcal{W}_N (t_j;0)^* \phi_+ (  h_{t_j} ) \mathcal{W}_N (t_j;0)}  \Omega \rangle \notag \\
=& \sum_{j=1}^m  \langle \Omega, \left( \prod_{k=1}^{j-1} e^{i\lambda_k \mathcal{W}_N (t_k;0)^*\phi_+ ( h_{t_k}) \mathcal{W}_N (t_k;0)} \right) \notag \\
& \hspace{1cm} \times \left(e^{i \lambda_j \mathcal{W}_N (t_j;0)^* ( \phi_+ ( h_{t_j} ) + N^{-1/2} \widetilde{A}_{t_j})\mathcal{W}_N (t_j;0)  }  -   e^{i \lambda_j\mathcal{W}_N (t_j;0)^*\phi_+ ( h_{t_j})\mathcal{W}_N (t_j;0)} \right) \notag \\
& \hspace{1cm} \times  \left( \prod_{\ell = j+1}^m e^{i \lambda_\ell \mathcal{W}_N (t_{\ell};0)^* ( \phi_+ ( h_{t_{\ell}} ) + N^{-1/2} \widetilde{A}_{t_{\ell}})\mathcal{W}_N (t_{\ell};0)  } \right) \Omega \rangle  \; . 
\end{align}
We write the difference as 
\begin{align}
\mathbb{E}  & \Big[ \prod_{j=1}^m e^{i \frac{\lambda_j}{\sqrt{N}} A_{t_j}} \Big]  - \langle  \Omega, \prod_{j=1}^m e^{i \lambda_j\mathcal{W}_N (t_j;0)^* \phi_+ (  h_{t_j} ) \mathcal{W}_N (t_j;0)}  \Omega \rangle \notag \\
=&  \int_0^1 ds \sum_{j=1}^m \lambda_j  \langle \Omega, \left( \prod_{k=1}^{j-1} e^{i \lambda_k\mathcal{W}_N (t_k;0)^* \phi_+ (h_{t_k})\mathcal{W}_N (t_k;0)}  \right) e^{i (1-s) \lambda_j\mathcal{W}_N (t_j;0)^*\phi_+ (h_{t_j})\mathcal{W}_N (t_j;0)} \notag \\
& \hspace{1cm} \times  N^{-1/2} \mathcal{W}_N (t_j;0)^*  \widetilde{A}_{t_j}\mathcal{W}_N (t_j;0)   e^{i s \lambda_j \mathcal{W}_N (t_j;0)^* ( \phi_+ ( h_{t_j} ) + N^{-1/2} B_{t_j}\mathcal{W}_N (t_j;0)  } \notag \\
& \hspace{1cm} \times  \left( \prod_{\ell = j+1}^m e^{i\lambda_\ell \mathcal{W}_N (t_{\ell};0)^* ( \phi_+ ( h_{t_{\ell}} ) + N^{-1/2} \widetilde{A}_{t_{\ell}})\mathcal{W}_N (t_{\ell};0)  }\right)  \Omega \rangle  \label{eq:step1-1}
\end{align}
We recall that $\widetilde{A}_{t_j} = d \Gamma ( q_{t_j} Aq_{t_j})$ and since $\|q_{t_j} Aq_{t_j} \| \leq C (1+ \| \varphi_{t_j} \|_2 ) \| A \| $ we find with Lemma \ref{lemma:WN} 
\begin{align}
\label{eq:B-approx}
\| \widetilde{A}_{t_j}\mathcal{W}_N (t_j;0) \psi \| \leq C \| A \| \| \mathcal{N} \mathcal{W}_N (t_j;0) \psi \| \leq C e^{C\vert t_ j \vert} \| ( \cN  + 1) \psi \|
\end{align}
by Lemma \eqref{lemma:WN} for any $\psi \in \mathcal{F}_{\perp \varphi_t}^{\leq N}$. Plugging this into \eqref{eq:step1-1} we get with Lemma \ref{lemma:WN} 
\begin{align}
\Big\vert \mathbb{E}  & \Big[ \prod_{j=1}^m e^{i \frac{\lambda_j}{\sqrt{N}} A_{t_j}} \Big]  - \langle  \Omega, \prod_{j=1}^m e^{i \lambda_j\mathcal{W}_N (t_j;0)^* \phi_+ (  h_{t_j} ) \mathcal{W}_N (t_j;0)}  \Omega \rangle \Big\vert \notag \\
\leq&  C N^{-1/2} \int_0^1 ds \sum_{j=1}^m  e^{C \vert t_j \vert }\lambda_j \notag \\
& \hspace{2cm} \times \| ( \mathcal{N} + 1) e^{i (1-s) \lambda_j\mathcal{W}_N (t_j;0)^*\phi_+ (h_{t_j})\mathcal{W}_N (t_j;0)}  \prod_{k=1}^{j-1} e^{i \lambda_k\mathcal{W}_N (t_k;0)^* \phi_+ (h_{t_k})\mathcal{W}_N (t_k;0)}    \Omega \| \notag \\
\leq& C N^{-1/2}\int_0^1 ds \sum_{j=1}^m \lambda_j e^{C \vert t_j \vert } \notag \\
& \hspace{2cm} \times \| ( \mathcal{N} + 1) e^{i (1-s) \lambda_j\phi_+ (h_{t_j})\mathcal{W}_N (t_j;0)}  \prod_{k=1}^{j-1} e^{i \lambda_k\mathcal{W}_N (t_k;0)^* \phi_+ (h_{t_k})\mathcal{W}_N (t_k;0)}    \Omega \| \; . 
\end{align}
From \cite[Lemma 3.5]{BSS} we get 
\begin{align}
\Big\vert \mathbb{E}  & \Big[ \prod_{j=1}^m e^{i \frac{\lambda_j}{\sqrt{N}} A_{t_j}} \Big]  - \langle  \Omega, \prod_{j=1}^m e^{i \lambda_j\mathcal{W}_N (t_j;0)^* \phi_+ (  h_{t_j} ) \mathcal{W}_N (t_j;0)}  \Omega \rangle \Big\vert \notag \\
\leq& C N^{-1/2}\int_0^1 ds \sum_{j=1}^m \lambda_j e^{C \vert t_j \vert } \notag \\
& \hspace{2cm} \times \| ( \mathcal{N} + (1-s)^2 \lambda_j^2 \| h_{t_j}\|^2 + 1)   \prod_{k=1}^{j-1} e^{i \lambda_k\mathcal{W}_N (t_k;0)^* \phi_+ (f_{t_k})\mathcal{W}_N (t_k;0)}    \Omega \| \; . 
\end{align}
Since $\| h_{t_j}\| = \| q_{t_j} A \varphi_{t_j} \| \leq \| A \| \| \varphi_{t_j} \|_2 \leq \| A \|$ we find after applying this bound to the remaining $j-1$ terms of the product 
\begin{align}
\Big\vert \mathbb{E}  & \Big[ \prod_{j=1}^m e^{i \frac{\lambda_j}{\sqrt{N}} A_{t_j}} \Big]  - \langle  \Omega, \prod_{j=1}^m e^{i \lambda_j\mathcal{W}_N (t_j;0)^* \phi_+ (  h_{t_j} ) \mathcal{W}_N (t_j;0)}  \Omega \rangle \Big\vert \notag \\
\leq& C N^{-1/2}\int_0^1 ds \sum_{j=1}^m \lambda_j e^{C \vert t_j \vert } \| ( \mathcal{N} + \sum_{k=1}^{j} \lambda_k^2 + 1)   \Omega \| \notag \\
&\leq CN^{-1/2} m \; e^{ \sum_{j=1}^m \vert t_j \vert} \sum_{j=1}^m \lambda_j^3 \; . \label{eq:step1}
\end{align}

\subsection*{Step 2} Next we compute the approximate action of the fluctuation dynamics on the observable $\phi_+ (h_{t_j})$, i.e. we are going to show that 
\begin{align}
\mathcal{W}_N (t_j;0)^* &  \phi_+ ( h_{t_j} )\mathcal{W}_N (t_j;0)  \approx \phi_+ (\mathcal{L}_{(t_j,0)}  h_{t_j}) 
\end{align}
where $\mathcal{L}_{(t_j, 0)}$ is given by \eqref{def:Ls}. To this end we define 
\begin{align}
\label{def:fs}
f_{(t_j,0)} :=  \mathcal{L}_{(t_j,0)}  h_{t_j} 
\end{align}
and compare 
\begin{align}
\langle  \Omega,  & \prod_{j=1}^m e^{i \lambda_j \phi_+ ( h_{t_j} ) }  \Omega \rangle  - \langle  \Omega, \prod_{j=1}^m e^{i \lambda_j \phi_+ ( f_{(t_j,0)} ) }  \Omega \rangle \notag \\
=& \sum_{j=1}^m  \langle \Omega, \left( \prod_{k=1}^{j-1} e^{i\lambda_k \phi_+ (f_{(t_k,0)})} \right) \left(e^{i \lambda_j \mathcal{W}_N (t_j;0)^*  \phi_+ ( h_{t_j} ) \mathcal{W}_N (t_j;0)  }  -   e^{i \lambda_j\phi_+ (f_{(t_j,0)})} \right) \notag \\
& \hspace{1cm} \times  \left( \prod_{\ell = j+1}^m e^{i \lambda_\ell \mathcal{W}_N (t_{\ell};0)^* \phi_+ ( h_{t_{\ell}} ) \mathcal{W}_N (t_{\ell};0)  } \right) \Omega \rangle  \notag \\
\end{align}
Similarly as in the first step we write this difference as 
\begin{align}
\langle  \Omega,  & \prod_{j=1}^m e^{i \lambda_j \phi_+ ( h_{t_j} ) }  \Omega \rangle - \langle  \Omega, \prod_{j=1}^m e^{i \lambda_j \phi_+ ( f_{(t_j,0)} ) }  \Omega \rangle \notag \\
=&  \int_0^1 ds \sum_{j=1}^m \lambda_j  \langle \Omega, \left( \prod_{k=1}^{j-1} e^{i \lambda_k \phi_+ (f_{(t_k,0)})}  \right) e^{i (1-s) \lambda_j\phi_+ (f_{(t_j,0)})} \notag \\
& \hspace{1cm} \times \left( \mathcal{W}_N (t_j;0)^*  \phi_+ ( h_{t_j} ) \mathcal{W}_N (t_j;0)  -\phi_+ (f_{(t_j,0)}) \right)  \notag \\
&\hspace{1cm} \times e^{i s \lambda_j \mathcal{W}_N (t_j;0)^* \phi_+ ( h_{t_j} ) \mathcal{W}_N (t_j;0)  } \notag \\
& \hspace{1cm} \times  \left( \prod_{\ell = j+1}^m e^{i\lambda_\ell \mathcal{W}_N (t_{\ell};0)^*  \phi_+ ( h_{t_{\ell}} ) \mathcal{W}_N (t_{\ell};0)  }\right)  \Omega \rangle  \; . \label{eq:step2-1}
\end{align}
To compute the difference, we recall for $s_j \in [0,t_j]$ the definiition of the function $f_{(t_j,s_j)}$ by \eqref{def:fs} 
with $f_{t_j,t_j}= h_{t_j}$. Then 
\begin{align}
\mathcal{W}_N (t_j;0)^*  & \phi_+ ( h_{t_j} ) \mathcal{W}_N (t_j;0)  -\phi_+ (f_{(t_j,0)})   \notag \\
=& -i \int_0^{t_j} ds_j \frac{id}{ds_j} \mathcal{W}_N (s_j;0)^*  \phi_+ ( f_{(t_j,s_j)} ) \mathcal{W}_N (s_j;0) \notag \\
=& i  \int_0^{t_j} ds_j  \mathcal{W}_N (s_j;0)^*\Big(  \Big[ \phi_+ ( f_{(t_j,s_j)}), \mathcal{L}_N (s_j) ) \Big]  - i\phi_- (i\partial_s f_{(t_j,s_j)}) \Big) \mathcal{W}_N (s_j;0)  \label{eq:comm-diff}
\end{align}
where we introduced the notation $i \phi (h) = b(h) - b^*(h)$ for any $h \in L^2( \mathbb{R}^3)$. With \eqref{eq:comm-b}, we compute the commutator and find a crucial cancellation between the leading order contribution of 
\begin{align}
\label{eq:commH-phi}
\Big[  \mathbb{H},  & \; \phi_+ ( f_{(t_j,s_j)} ) \Big] \notag \\
=& i\phi_- ( (h_{\varphi_s} + \widetilde{K}_{1,s} - \widetilde{K}_{2,s}) f_{(t_j,s_j)} )  \notag \\
&+ \frac{1}{2} b^*(K_{2,s_j} J f_{(t_j,s_j)}) \frac{\cN }{N} + \frac{1}{2} \frac{\cN }{N}  b^*(K_{2,s_j} J f_{(t_j,s_j)})  - {\rm h.c.} \notag \\
&+ \frac{1}{2N} \int dxdy K_{2,s_j}(x,y)  a( f_{(t_j,s_j)})  a_x^* b_y^* + \frac{1}{2N} \int dxdy K_{2,s_j}(x,y) b_x^*  a_y^*a( f_{(t_j,s_j)}) \notag \\
&- {\rm h.c.}
\end{align}
and $ i\phi_- (i\partial_s f_{(t_j,s_j)}) $. Recall that $\|K_{2,s_j} \|_2 \leq C $ and $\| K_{2,s_j} J f_{(t_j,s_j)}\|_2 \leq C \|  f_{(t_j,s_j)} \|_ \leq C e^{C \vert t_j \vert}$. Thus the three last lines of the r.h.s. of \eqref{eq:commH-phi} do not contribute to leading order and 
\begin{align}
\| ( \Big[ \mathbb{H},  &  \; \phi_+ ( f_{(t_j,s_j)} ) \Big] - i\phi_i (\partial_s f_{(t_j,s_j)})  ) \psi \| \leq C e^{C \vert t_j \vert}N^{-1} \|( \cN  + 1)^{3/2} \psi \| \; .\label{eq:comm-H}
\end{align}
Next we show that the commutators with the remainders $\mathcal{R}_N^{(j)}$ are negligible. For this we consider the single contributions of $\mathcal{R}_N = \mathcal{R}_N^{(1)} + \mathcal{R}_N^{(2)}$ separately and start with the first one 
\begin{align}
\Big[ \mathcal{R}_{N,s_j}^{(1)},  \phi_+ ( f_{(t_j,s_j)} ) \Big] =& \frac{1}{2} i\phi_- (q_{s_j} [V*\vert \varphi_{s_j}\vert^2 + K_{1,s_j} - \mu_{s_j}]  f_{(t_j,s_j)} )\frac{1-\cN }{N} \notag \\
&- \frac{1}{N} d\Gamma ( q_{s_j} [v*\vert \varphi_{s_j}\vert^2 + K_{1,s_j} - \mu_{s_j}]q_{s_j}) i\phi_-( f_{(t_j,s_j)} ) 
\notag \\
&-{\rm h.c.}
\end{align}
Since $\| [v*\vert \varphi_{s_j}\vert^2 + K_{1,s_j} - \mu_{s_j}] \|_\infty \leq C$, we thus get 
\begin{align}
\| \Big[ \mathcal{R}_{N,s_j}^{(1)},  \phi_+ ( f_{(t_j,s_j)} ) \Big] \psi \| \leq  C e^{C \vert t_j \vert}N^{-1} \|( \cN  + 1)^{2} \psi \| \; .|\label{eq:comm-R1}
\end{align}
Moreover, 
\begin{align}
\Big[ \mathcal{R}_{N,s_j}^{(2)},  \phi_+ ( f_{(t_j,s_j)} ) \Big] =& \frac{1}{\sqrt{N}}\int v(x-y)\varphi_{s_j} (y) \overline{f}_{t_j,s_j} (x)  b^*(q_{s_j,x}) b(q_{t,y}) dxdy \notag \\
&- \frac{1}{\sqrt{N}}\int v(x-y)\varphi_{s_j} (y) f_{(t_j,s_j)} (x)  b(q_{s_j,x}) b(q_{s_j,y}) dxdy \notag \\
&-\frac{1}{\sqrt{N}}\int v(x-y)\varphi_{s_j} (y) f_{(t_j,s_j)} (y)  a^*(q_{s_j,x})a(q_{s_j,x}) ( 1- \frac{\cN }{N}) dxdy \notag \\
&-\frac{1}{N^{3/2}}\int v(x-y)\varphi_{s_j} (y)  a^*(q_{s_j,x})a(q_{s_j,x}) a^*(f_{(t_j,s_j)} ) b(q_{s_j,y})\notag \\
&- {\rm h.c.}
 \end{align}
Using that 
\begin{align}
\int dxdy v^2(x-y) \vert \varphi_t (y)\vert^2  \vert \overline{f}_{(t_j,s_j)} (x) \vert^2 \leq \| v^2 * \vert \varphi_t \vert^2 \|_\infty \| \overline{f}_{(t_j,s_j)} \|_2 \leq C e^{\vert t_j \vert}
\end{align}
for the first three lines and 
\begin{align}
\vert \langle &  \xi, (\cN  + 1)^{-1/2} \int v(x-y)\varphi_{s_j} (y)  a^*(q_{s_j,x})a(q_{s_j,x}) a^*(f_{(t_j,s_j)} ) b(q_{s_j,y})(\cN  + 2)^{1/2}  \psi \rangle \vert  \notag \\
\leq& \left( \int dx (  v^2 * \vert \varphi_{s_j} \vert^2 )(x)  \vert  \| a(q_{s_j,x}) (\cN  + 1)^{-1/2}  \xi \|^2 \right)^{1/2} \notag \\
& \hspace{2cm} \times \left( \int dxdy  \|a(q_{s_j,x}) a^*(f_{(t_j,s_j)} ) b(q_{s_j,y}) (\cN  + 1)^{1/2}  \psi \|^2 \right)^{1/2} \notag \\
&\leq  C  e^{\vert t_j \vert} \| \xi \| ( \cN  + 1)^2 \psi \| \label{eq:comm-R3-3}
\end{align}
for the last line we arrive at 
\begin{align}
\| \Big[ \mathcal{R}_{N,s_j}^{(2)},  \phi_+ ( f_{(t_j,s_j)} ) \Big]  \psi \| \leq CN^{-1/2} e^{\vert t_j \vert} \| \xi \| \; \| ( \cN  + 1)^2 \psi \| \; . \label{eq:comm-R2}
\end{align}
For the last contribution of the remainder we find 
\begin{align}
\Big[ & \mathcal{R}_{N,s_j}^{(3)},  \phi_+ ( f_{(t_j,s_j)} ) \Big]  \notag \\
=& \frac{1}{N}\int dxdy v(x-y) \left( \overline{f}_{t_j,s_j} (y) b_y^*a_x^*a_x - f_{(t_j,s_j)} (y) a_x^*a_xb_y \right)  \; . 
\end{align}
Then similarly as in \eqref{eq:comm-R3-3} with $\| f_{(t_j,s_j)} \|_{H^2} \leq e^{e^{C \vert t_j \vert}}$ we find that 
\begin{align}
\| \Big[ & \mathcal{R}_{N,s_j}^{(3)},  \phi_+ ( f_{(t_j,s_j)} ) \Big]  \psi \| \leq N^{-1} e^{e^{C \vert t_j \vert}} \| ( \cN  + 1)^2 \psi \| \; . \label{eq:comm-R3}
\end{align}
Summarizing \eqref{eq:comm-H},\eqref{eq:comm-R1},\eqref{eq:comm-R2} and \eqref{eq:comm-R3} we thus arrive for any $\psi \in \mathcal{F}_{\perp \varphi_t}^{\leq N}$ at 
\begin{align}
\| \Big(  \Big[ \mathcal{L}_N (s_j),  \phi_+ ( f_{(t_j,s_j)} ) \Big]  - i\phi_i (\partial_s f_{(t_j,s_j)}) \Big)  \psi \| \leq N^{-1/2} e^{e^{C \vert t_j \vert}} \| ( \cN  + 1)^{3/2} \psi \| \; . 
\end{align}
Plugging this into \eqref{eq:comm-diff} we get with Lemma \ref{lemma:WN}
\begin{align}
\| (\mathcal{W}_N (t_j;0)^*  & \phi_+ ( h_{t_j} ) \mathcal{W}_N (t_j;0)  -\phi_+ (f_{(t_j,s_j))} ) \psi \| \leq N^{-1/2} e^{e^{C \vert t_j \vert}} \| ( \cN  + 1)^{3/2} \psi \|\label{eq:aprox-action}
\end{align}
that we use now to estimate \eqref{eq:step2-1} by 
\begin{align}
\big\vert \langle  \Omega,  & \prod_{j=1}^m e^{i \lambda_j \phi_+ ( h_{t_j} ) }  \Omega \rangle - \langle  \Omega, \prod_{j=1}^m e^{i \lambda_j \phi_+ ( f_{(t_j,s_j)} ) }  \Omega \rangle \big\vert \notag \\
\leq&  N^{-1/2}  \int_0^1 ds \sum_{j=1}^m \lambda_j   e^{e^{C \vert t_j \vert}} \| ( \cN  + 1)^{3/2} e^{i (1-s) \lambda_j\phi_+ (f_{(t_j,s_j)})}  \left( \prod_{k=1}^{j-1} e^{i \lambda_k \phi_+ (f_{(t_k,s_k)})}  \right) \Omega \|  \; .
\end{align}
As in step 1 we conclude by \cite[Lemma 3.5]{BSS} that 
\begin{align}
\big\vert \langle  \Omega,  & \prod_{j=1}^m e^{i \lambda_j \phi_+ ( h_{t_j} ) }  \Omega \rangle - \langle  \Omega, \prod_{j=1}^m e^{i \lambda_j \phi_+ ( f_{(t_j,s_j)} ) }  \Omega \rangle \big\vert 
\leq C  N^{-1/2}   m e^{\sum_{j=1}^m \vert t_j \vert} \sum_{j=1}^m \lambda_j^3 \; . \label{eq:step2}
\end{align}

\subsection*{Step 3} In this step we show that we can replace the modified with the standard creation and annihilation operators, i.e. that 
\begin{align}
\phi_+ ( f_{(t_j,s_j)} ) \approx \widetilde{\phi}_+ ( f_{(t_j,s_j)} )
\end{align}
where $\widetilde{\phi}_+ ( f_{(t_j,s_j)} ) = a(f_{(t_j,s_j)}) + a^*(f_{(t_j,s_j)})$. Note that $\widetilde{\phi}_+$ is defined on the full bosonic Fock space, but $\phi_+$ on the truncated Fock space only. For this reason we 
\begin{align}
 \langle  \Omega,  & \prod_{j=1}^m e^{i \lambda_j \phi_+ ( f_{(t_j,s_j)} ) }  \Omega \rangle - \langle  \Omega, \prod_{j=1}^m e^{i \lambda_j \widetilde{\phi}_+ ( f_{(t_j,s_j)} ) }  \Omega \rangle \notag \\
=&  \langle  \Omega,   \prod_{j=1}^m e^{i \lambda_j \phi_+ ( f_{(t_j,s_j)} ) }  \Omega \rangle - \langle  \Omega, \prod_{j=1}^m e^{i \lambda_j\mathds{1}_{\mathcal{N} \leq N} \widetilde{\phi}_+ ( f_{(t_j,s_j)} ) \mathds{1}_{\mathcal{N} \leq N} } \Omega \rangle  \notag \\
&+ \langle  \Omega, \prod_{j=1}^m e^{i \lambda_j\mathds{1}_{\mathcal{N} \leq N} \widetilde{\phi}_+ ( f_{(t_j,s_j)} ) \mathds{1}_{\mathcal{N} \leq N} } \Omega \rangle - \langle  \Omega, \prod_{j=1}^m e^{i \lambda_j  \widetilde{\phi}_+ ( f_{(t_j,s_j)} ) } \Omega \rangle \label{eq:step30}
\end{align}
and show that the difference in the first resp. the difference in the second line are small. 
For the first, we proceed as in the previous steps  
\begin{align}
\langle  \Omega,  & \prod_{j=1}^m e^{i \lambda_j \phi_+ ( f_{(t_j,s_j)} ) }  \Omega \rangle - \langle  \Omega, \prod_{j=1}^m e^{i \lambda_j \mathds{1}_{\mathcal{N} \leq N} \widetilde{\phi}_+ ( f_{(t_j,s_j)} ) \mathds{1}_{\mathcal{N} } \leq N} \Omega \rangle \notag \\
=&  \int_0^1 ds \sum_{j=1}^m \lambda_j  \langle \Omega,  \left( \prod_{k=1}^{j-1} e^{i \lambda_k \phi_+ (f_{t_k,s_k})}  \right) e^{i (1-s) \lambda_j\phi_+ (f_{(t_j,s_j)})} \notag \\
& \hspace{1cm} \times \mathds{1}_{\mathcal{N} \leq N} \left(   \phi_+ ( f_{(t_j,s_j)} )  -\widetilde{\phi}_+ (f_{(t_j,s_j)}) \right) \mathds{1}_{\mathcal{N} \leq N}  e^{i s \lambda_j \widetilde{\phi}_+ ( f_{(t_j,s_j)} ) \mathds{1}_{\mathcal{N} \leq N}  }   \notag \\
& \hspace{7cm} \times \left( \prod_{\ell = j+1}^m e^{i\lambda_\ell  \mathds{1}_{\mathcal{N} \leq N} \widetilde{\phi}_+ ( f_{t_{\ell,s_\ell}} ) \mathds{1}_{\mathcal{N} \leq N} }\right) \Omega \rangle  \; .
\end{align}
Since 
\begin{align}
\label{eq:a-b}
\| \left(   \phi_+ ( f_{(t_j,s_j)} )  -\widetilde{\phi}_+ (f_{(t_j,s_j)}) \right)  \mathds{1}_{\mathcal{N} \leq N} \psi \| \leq& \| ( \sqrt{1- \cN /N} - 1) a((f_{(t_j,s_j)}) \mathds{1}_{\mathcal{N} \leq N}  \psi \| \notag \\
&+ \|a^*((f_{(t_j,s_j)})  ( \sqrt{1- \cN /N} - 1) \mathds{1}_{\mathcal{N} \leq N} \psi \|\notag \\
\leq& C N^{-1} e^{C \vert t_j \vert} \| ( \cN  + 1)^{3/2} \psi \| 
\end{align}
we conclude by \cite[Lemma 3.5]{BSS} that 
\begin{align}
\big\vert \langle  \Omega, &  \prod_{j=1}^m e^{i \lambda_j \phi_+ ( f_{(t_j,s_j)} ) }  \Omega \rangle - \langle  \Omega,\mathds{1}_{\mathcal{N} \leq N}  \prod_{j=1}^m e^{i \lambda_j \mathds{1}_{\mathcal{N} \leq N} \widetilde{\phi}_+ ( f_{(t_j,s_j)} )\mathds{1}_{\mathcal{N} \leq N} } \mathds{1}_{\mathcal{N} \leq N}  \Omega \rangle \big\vert \notag \\
&\leq C  N^{-1}   m e^{\sum_{j=1}^m \vert t_j \vert} \sum_{j=1}^m \lambda_j^3 \; . \label{eq:step3-1}  
\end{align}
For the second line of the r.h.s. of \eqref{eq:step30} we write 
\begin{align}
    \langle  \Omega, & \prod_{j=1}^m e^{i \lambda_j\mathds{1}_{\mathcal{N} \leq N} \widetilde{\phi}_+ ( f_{(t_j,s_j)} ) \mathds{1}_{\mathcal{N} \leq N} } \Omega \rangle - \langle  \Omega, \prod_{j=1}^m e^{i \lambda_j  \widetilde{\phi}_+ ( f_{(t_j,s_j)} ) } \Omega \rangle  \notag \\
     =& \int_0^t ds \sum_{j=1}^m \lambda_j \langle \Omega, \bigg( \prod_{k=1}^{j-1} e^{i \lambda_k \mathds{1}_{\mathcal{N} \leq N}\widetilde{\phi}_+ (f_{t_k,s_k} ) \mathds{1}_{\mathcal{N} \leq N}} \bigg) e^{i (1-s) \mathds{1}_{\mathcal{N} \leq N}\lambda_j \widetilde{\phi}_+ (f_{t_j,s_j})\mathds{1}_{\mathcal{N} \leq N}} \notag \\
   &  \hspace{1cm} \times \bigg( \mathds{1}_{\mathcal{N} \leq N}  \widetilde{\phi}_+ ( f_{(t_j,s_j)} )  \mathds{1}_{\mathcal{N} \leq N} -  \widetilde{\phi}_+ ( f_{(t_j,s_j)} ) \bigg)  e^{i s \lambda_j \widetilde{\phi}_+ ( f_{(t_j,s_j)} ) \mathds{1}_{\mathcal{N} \leq N}  }   \left( \prod_{\ell = j+1}^m e^{i\lambda_\ell   \widetilde{\phi}_+ ( f_{t_{\ell},s_\ell} )}\right) \Omega \rangle  \; .
\end{align}
Since 
\begin{align}
    \mathds{1}_{\mathcal{N} \leq N}  \widetilde{\phi}_+ ( f_{(t_j,s_j)} )  \mathds{1}_{\mathcal{N} \leq N} -  \widetilde{\phi}_+ ( f_{(t_j,s_j)} )  =  \mathds{1}_{\mathcal{N} > N}   \widetilde{\phi}_+ ( f_{(t_j,s_j)} ) +  \mathds{1}_{\mathcal{N} \leq  N} \widetilde{\phi}_+ ( f_{(t_j,s_j)} ) \mathds{1}_{\mathcal{N} > N} 
\end{align}
and $\mathds{1}_{\mathcal{N} > N} \leq \mathcal{N}/N$, we find 
\begin{align}
    \| \mathds{1}_{\mathcal{N} \leq N}  \widetilde{\phi}_+ ( f_{(t_j,s_j)} )  \mathds{1}_{\mathcal{N} \leq N} -  \widetilde{\phi}_+ ( f_{(t_j,s_j)} ) \psi \| \leq C N^{-1} e^{C \vert t_j \vert} \|( \mathcal{N} + 1)^{3/2} \psi \| \label{eq:diff-tildephi}
\end{align}
and we arrive with \cite[Lemma 3.5]{BSS} at 
\begin{align}
    \big\vert  \langle  \Omega, & \prod_{j=1}^m e^{i \lambda_j\mathds{1}_{\mathcal{N} \leq N} \widetilde{\phi}_+ ( f_{(t_j,s_j)} ) \mathds{1}_{\mathcal{N} \leq N} } \Omega \rangle - \langle  \Omega, \prod_{j=1}^m e^{i \lambda_j  \widetilde{\phi}_+ ( f_{(t_j,s_j)} ) } \Omega \rangle  \big\vert \leq  C  N^{-1}   m e^{\sum_{j=1}^m \vert t_j \vert} \sum_{j=1}^m \lambda_j^3 
\end{align}
that, together with \eqref{eq:step3-1} shows 
\begin{align}
    \big\vert  \langle  \Omega, & \prod_{j=1}^m e^{i \lambda_j \widetilde{\phi}_+ ( f_{(t_j,s_j)} ) } \Omega \rangle - \langle  \Omega, \prod_{j=1}^m e^{i \lambda_j  \phi_+ ( f_{(t_j,s_j)} ) } \Omega \rangle  \big\vert \leq  C  N^{-1}   m e^{\sum_{j=1}^m \vert t_j \vert} \sum_{j=1}^m \lambda_j^3  \; . \label{eq:step3}
\end{align}

\subsection*{Step 4} We compute the remaining expectation value with Baker Campbell Hausdorff 
\begin{align}
\langle  \Omega, & \prod_{j=1}^m e^{i \lambda_j \widetilde{\phi}_+ ( f_{(t_j,s_j)} ) }  \Omega \rangle \notag \\
&= e^{- \sum_{j=1}^m  \lambda_j^2 \| f_{(t_j,s_j)} \|^2/2}\langle  \Omega, e^{i  a^* ( f_{(t_1,s_1)} ) } e^{i  a ( f_{(t_1,s_1)} ) } \cdots e^{i  a^* ( f_{(t_m,s_m)} ) } e^{i  a ( f_{(t_m,s_m)} ) }  \Omega \rangle \notag \\
&= e^{- \sum_{j=1}^m \| f_{(t_j,s_j)} \|^2/2}\langle  \Omega,  e^{i  a ( f_{(t_1,s_1)} ) } \prod_{j=2}^{m-1} e^{i  a^* ( f_{(t_j,s_j)} ) } e^{i  a ( f_{(t_j,s_j)} ) }  e^{i  a^* ( f_{(t_m,s_m)} ) }  \Omega \rangle 
\end{align}
Using again Baker Campbell Hausdorff we find 
\begin{align}
\langle  & \Omega, e^{i  \phi_+ ( f_{t_1} ) } \cdots e^{i  \phi_+ (f_m )} \Omega \rangle  \notag \\
=& e^{- \sum_{j=1}^m \| f_{(t_j,s_j)} \|^2/2 - \sum_{j=2}^m \lambda_1 \lambda_j \langle f_{(t_1,s_1)}, f_{(t_j,s_j)}\rangle} \langle  \Omega,  e^{i  a ( f_{(t_2,s_2)} ) } \prod_{j=3}^{m-1} e^{i  a^* ( f_{(t_j,s_j)} ) } e^{i  a ( f_{(t_j,s_j)} ) }  e^{i  a^* ( f_{(t_m,s_m)} ) }  \Omega \rangle \notag \\
=& e^{- \sum_{j=1}^m \| f_{(t_j,s_j)} \|^2/2 - \sum_{k=1}^{m-1} \sum_{j=k+1}^m \lambda_j \lambda_k \langle f_{(t_k,s_k)}, f_{(t_j,s_j)}\rangle} \notag\\
=& e^{- \sum_{i,j =1}^m \Sigma_{i,j} \lambda_i \lambda_j/2}\label{eq:step4}
\end{align}
where $\Sigma \in \mathbb{R}^{m \times m}$ is given by \eqref{def:Sigma}. 

Combining \eqref{eq:step1},\eqref{eq:step2}, \eqref{eq:step3} and \eqref{eq:step4}, Theorem \ref{thm:MCLT} follows. 
\end{proof}

\section{Proof of Theorem \ref{thm:corrfct}}
\label{sec:corr}

As mentioned in the introduction, the ideas of the proof of Theorem \ref{thm:corrfct} are similar to those of the Theorem \ref{thm:MCLT}.  

\begin{proof}[Proof of Theorem \ref{thm:corrfct}]
Accordingly, we rewrite the correlation function in the bosonic Fock space 
\begin{align}
\mathbb{E} \big[ \prod_{i=1}^m (\mathcal{A}_{t_i}/ \sqrt{N} )\big] = \langle \psi_N, \prod_{i=1}^m d\Gamma ( A_{t_i} ) / \sqrt{N } \psi_N \rangle  
\end{align}
in the first step using the excitation map \eqref{def:U} as 
\begin{align}
\mathbb{E} \big[ \prod_{i=1}^m (\mathcal{A}_{t_i}/ \sqrt{N} )\big]  = \langle \mathcal{U}_{N,0}^* \psi_N, \prod_{i=1}^m \left( \mathcal{U}_{N,0}^* d\Gamma ( A_{t_i} )   \mathcal{U}_{N,0} / \sqrt{N} \right)  \mathcal{U}_{N,0}^* \psi_N \rangle  \; . 
\end{align}
With the definition of the fluctuation map \eqref{def:flucdyn} we can rewrite the observables, similarly as in \eqref{eq:B} as 
\begin{align}
 \mathcal{U}_{N,0}^* N^{-1/2} d\Gamma ( A_{t_i} ) \mathcal{U}_{N,0}  \mathcal{U}_{N,0}^* = \mathcal{W}_N(t_i;0 )^* \left( \phi_+ (q_{t_i} A \varphi_{t_i} ) + N^{-1/2} \widetilde{A}_{t_i} \right) \mathcal{W}_N(t_i;0 ) 
\end{align}
where we set $\widetilde{A}_{t_i} = d \Gamma \left( q_{t_i} (A-\langle \varphi_{t_i}, \;  A \varphi_{t_i} \rangle )q_{t_i}\right)$. Since $\mathcal{U}_{N,0}^*\psi_N = \Omega$ we are thus left with computing the expectation value 
\begin{align}
\mathbb{E} \big[ \prod_{i=1}^m (\mathcal{A}_{t_i}/ \sqrt{N} )\big]  = \langle \Omega, \prod_{i=1}^m \left( \mathcal{W}_N(t_i;0 )^* \left( \phi_+ (q_{t_i} A \varphi_{t_i} ) + N^{-1/2} \widetilde{A}_{t_i} \right) \mathcal{W}_N(t_i;0 )  \right) \Omega \rangle \;  
\end{align}
for which we follow the lines of the proof of Theorem \ref{thm:MCLT}.  

\subsection*{Step 1} As in the first step of the proof of Theorem \ref{thm:MCLT}, we start with showing that the operator $N^{-1/2} \widetilde{A}_{t_i}$ is negligible. We write 
\begin{align}
 \langle \Omega,  & \prod_{i=1}^m \left( \mathcal{W}_N(t_i;0 )^* \left( \phi_+ (q_{t_i} A \varphi_{t_i} ) + N^{-1/2} \widetilde{A}_{t_i} \right) \mathcal{W}_N(t_i;0 )  \right) \Omega \rangle \notag \\
 &  - \langle \Omega, \prod_{i=1}^m \left( \mathcal{W}_N(t_i;0 )^*  \phi_+ (q_{t_i} A \varphi_{t_i} )   \mathcal{W}_N(t_i;0 )  \right) \Omega \rangle  \notag \\
 & = N^{-1/2} \sum_{k=1}^m \langle \Omega, \prod_{i=1}^{k-1}  \mathcal{W}_N(t_i;0 )^* \left( \phi_+ (q_{t_i} A \varphi_{t_i} ) + N^{-1/2} \widetilde{A}_{t_i} \right) \mathcal{W}_N(t_i;0 )^* \notag \\
 & \hspace{2cm}  \times \mathcal{W}_N(t_k;0 )^* \widetilde{A}_{t_k} \mathcal{W}_N(t_k;0 )^*  \; \prod_{j=k+1}^{m}  \mathcal{W}_N(t_j;0 )^*  \phi_+ (q_{t_j} A \varphi_{t_j} )  \mathcal{W}_N(t_j;0 )^* \Omega \rangle \; . 
\end{align}
On the one hand, from \eqref{eq:B-approx} and Lemma \ref{lemma:WN} we have 
\begin{align}
\| \widetilde{A}_{t_k} & \mathcal{W}_N(t_k;0 )^*  \; \prod_{j=k+1}^{m}  \mathcal{W}_N(t_j;0 )^*  \phi_+ (q_{t_j} A \varphi_{t_j} )  \mathcal{W}_N(t_j;0 )^* \Omega \| \notag \\
& \leq C \| ( \mathcal{N} + 1) \mathcal{W}_N(t_k;0 )^*   \prod_{j=k+1}^{m}  \mathcal{W}_N(t_j;0 )^*  \phi_+ (q_{t_j} A \varphi_{t_j} )  \mathcal{W}_N(t_j;0 )^* \Omega \| \; \notag \\
&\leq C e^{C \vert t_k \vert + \vert t_{k+1}\vert } \notag \\
& \; \times  \| ( \mathcal{N} + 1) \phi_+ (q_{t_{k+1}} A \varphi_{t_{k+1}} )  \mathcal{W}_N(t_{k+1};0 )^* \prod_{j=k+2}^{m}  \mathcal{W}_N(t_j;0 )^*  \phi_+ (q_{t_j} A \varphi_{t_j} )  \mathcal{W}_N(t_j;0 )^* \Omega \| \;. 
\end{align} 
Furthermore for any $\psi \in \mathcal{F}^{\leq N}$ we have from the commutation relations 
\begin{align}
\| ( \mathcal{N} + 1) \phi_+ (h) \psi \| &\leq \| ( \mathcal{N} + 1) b^*(h) \psi \| + \| ( \mathcal{N}+  1) b(h) \psi \| \notag \\
&= \|  b^*(h) ( \mathcal{N} + 2)\psi  \| + \|  b(h)  \mathcal{N} \psi \| \notag \\
&\leq C \| h \|_2 \| ( \mathcal{N} + 2)^{3/2} \psi \| \; \label{eq:N-phi+}
\end{align}
that leads with Lemma \ref{lemma:WN} to 
\begin{align}
\| \widetilde{A}_{t_k} & \mathcal{W}_N(t_k;0 )^*  \; \prod_{j=k+1}^{m}  \mathcal{W}_N(t_j;0 )^*  \phi_+ (q_{t_j} A \varphi_{t_j} )  \mathcal{W}_N(t_j;0 )^* \Omega \| \notag \\
&\leq C e^{C \vert t_k \vert + \vert t_{k+1}\vert } \| ( \mathcal{N} + 1)^{3/2} \prod_{j=k+2}^{m}  \mathcal{W}_N(t_j;0 )^*  \phi_+ (q_{t_j} A \varphi_{t_j} )  \mathcal{W}_N(t_j;0 )^* \Omega \| \;. 
\end{align} 
We use \eqref{eq:N-phi+} and Lemma \ref{lemma:WN} iteratively for all the factors of the product and end up with 
\begin{align}
\| \widetilde{A}_{t_k} & \mathcal{W}_N(t_k;0 )^*  \; \prod_{j=k+1}^{m}  \mathcal{W}_N(t_j;0 )^*  \phi_+ (q_{t_j} A \varphi_{t_j} )  \mathcal{W}_N(t_j;0 )^* \Omega \| \notag \\
& \leq C_m e^{C \sum_{i=k}^m \vert t_i \vert } \| ( \mathcal{N} + 1 )^{(m-k+1)/2} \Omega \| \leq  C_m e^{C \sum_{i=k}^m \vert t_i \vert }  \; . 
\end{align}
On the other hand the same arguments show 
\begin{align}
\| \prod_{i=1}^{k-1}  \mathcal{W}_N(t_i;0 )^* \left( \phi_+ (q_{t_i} A \varphi_{t_i} ) + N^{-1/2} \widetilde{A}_{t_i} \right) \mathcal{W}_N(t_i;0 )^*\Omega \| \leq C_k e^{\sum_{i=1}^{k-1} \vert t_i \vert} \; . 
\end{align}
Thus with Cauchy Schwarz 
\begin{align}
\big\vert & \langle \Omega,   \prod_{i=1}^m \left( \mathcal{W}_N(t_i;0 )^* \left( \phi_+ (q_{t_i} A \varphi_{t_i} ) + N^{-1/2} \widetilde{A}_{t_i} \right) \mathcal{W}_N(t_i;0 )  \right) \Omega \rangle \notag \\
 & \hspace{1cm}  - \langle \Omega, \prod_{i=1}^m \left( \mathcal{W}_N(t_i;0 )^*  \phi_+ (q_{t_i} A \varphi_{t_i} )   \mathcal{W}_N(t_i;0 )  \right) \Omega \rangle  \big\vert \notag  
 \leq C_m N^{-1/2} e^{\sum_{i=1}^{m } \vert t_i \vert} \; . 
\end{align}

\subsection*{Step 2} In the second step we compute the approximate action of the fluctuation dynamics $\mathcal{W}_N(t_i, 0)$ on the observable $\phi_+ (q_{t_i} A \varphi_{t_i})$. For this we recall the definition $f_{(t_i0)} = \mathcal{L}_{(t_i; 0)} h_{t_i}$ with $h_{t_i} =q_{t_i} A \varphi_{t_i} $ that lead to 
\begin{align}
\langle \Omega, & \prod_{i=1}^m  \mathcal{W}_N(t_i;0 )^*  \phi_+ (q_{t_i} A \varphi_{t_i} )   \mathcal{W}_N(t_i;0 )   \Omega \rangle -\langle \Omega,  \prod_{i=1}^m  \phi_+ (f_{(t_i;0)} ) \Omega \rangle \notag \\
=& \sum_{k=1}^m \langle \Omega,  \prod_{i=1}^{k-1}  \mathcal{W}_N(t_i;0 )^*  \phi_+ (h_{t_i} )   \mathcal{W}_N(t_i;0 ) \notag \\
& \hspace{2cm} \times \left(  \mathcal{W}_N(t_k;0 )^*  \phi_+ (h_{t_k} )   \mathcal{W}_N(t_k;0 ) -  \phi_+ (f_{(t_k;0)} \right)  \prod_{j=k+1}^m  \phi_+ (f_{(t_j;0)} ) \Omega \rangle \; . 
\end{align}
From \eqref{eq:aprox-action} we find 
\begin{align}
\|&  \left(  \mathcal{W}_N(t_k;0 )^*  \phi_+ (h_{t_k} )   \mathcal{W}_N(t_k;0 ) -  \phi_+ (f_{(t_k;0)} \right)  \prod_{j=k+1}^m  \phi_+ (f_{(t_j;0)} ) \Omega  \| \notag \\
&\leq C N^{-1/2} e^{ e^{C \vert t_k \vert }} \| (\mathcal{N} + 1)^{3/2}  \prod_{j=k+1}^m  \phi_+ (f_{(t_j;0)} ) \Omega \| \; . 
\end{align}
Since $\| f_{(t_j;0)} \|_2 \leq C e^{C \vert t_j \vert } $ from Lemma \ref{lemma:f}, we thus find using again \eqref{eq:N-phi+} 
\begin{align}
\| & \left(  \mathcal{W}_N(t_k;0 )^*  \phi_+ (h_{t_k} )   \mathcal{W}_N(t_k;0 ) -  \phi_+ (f_{(t_k;0)} \right)  \prod_{j=k+1}^m  \phi_+ (f_{(t_j;0)} ) \Omega  \| \notag \\
&\leq C_m N^{-1/2} e^{ e^{C \vert t_k \vert }} e^{\sum_{i=k+1}^m \vert t_i \vert }  \; . 
\end{align}
Similar calculations as in the first step show 
\begin{align}
\| \prod_{i=1}^{k-1} \mathcal{W}_N(t_i;0 )^*  \phi_+ (h_{t_i} )   \mathcal{W}_N(t_i;0 ) \Omega \| \leq C_k e^{C \sum_{i=1}^{k-1} \vert t_i \vert}   
\end{align}
that lead to 
 \begin{align}
 \vert \langle \Omega, & \prod_{i=1}^m  \mathcal{W}_N(t_i;0 )^*  \phi_+ (q_{t_i} A \varphi_{t_i} )   \mathcal{W}_N(t_i;0 )   \Omega \rangle -\langle \Omega,  \prod_{i=1}^m  \phi_+ (f_{(t_i;0)} ) \Omega \rangle \vert  \notag \\
 \leq& C_m  N^{-1/2} e^{C \sum_{i=1}^m e^{C \vert t_i \vert }} \; . 
 \end{align}
 
\subsection*{Step 3} In the third step we replace the modified creation and annihilation operators with standard ones. We recall the definition $\widetilde{\phi}_+ (h) = a(h) + a^*(h)$ for all $h \in L^2( \mathbb{R}^3)$ and write (with similar arguments as in \eqref{eq:step30}
\begin{align}
\langle \Omega, & \prod_{i=1}^m   \phi_+ ( f_{(t_i;0)} )   \Omega \rangle -\langle \Omega,  \prod_{i=1}^m  \widetilde{\phi}_+ (f_{(t_i;0)} ) \Omega \rangle \notag \\
=&  \langle \Omega,  \prod_{i=1}^m   \phi_+ ( f_{(t_i;0)} )   \Omega \rangle -\langle \Omega,  \prod_{i=1}^m \big( \mathds{1}_{\mathcal{N} \leq N} \widetilde{\phi}_+ (f_{(t_i;0)} ) \mathds{1}_{\mathcal{N} \leq N} \big) \Omega \rangle \notag \\
&+ \langle \Omega,  \prod_{i=1}^m \big( \mathds{1}_{\mathcal{N} \leq N} \widetilde{\phi}_+ (f_{(t_i;0)} ) \mathds{1}_{\mathcal{N} \leq N} \big) \Omega \rangle - \langle \Omega,  \prod_{i=1}^m   \widetilde{\phi}_+ ( f_{(t_i;0)} )   \Omega \rangle \label{eq:step301}
\end{align}
and show that the differences in the first and second line of the r.h.s. of \eqref{eq:step301} are small. For the terms in the first line we have 
\begin{align}
\langle \Omega, & \prod_{i=1}^m   \phi_+ ( f_{(t_i;0)} )   \Omega \rangle -\langle \Omega,  \prod_{i=1}^m  \big( \mathds{1}_{\mathcal{N} \leq N} \widetilde{\phi}_+ (f_{(t_i;0)} ) \mathds{1}_{\mathcal{N} \leq N} \big)  \Omega \rangle \notag \\
=& \sum_{k=1}^m \langle \Omega,  \prod_{i=1}^{k-1}    \phi_+ ( f_{(t_i;0)} ) \mathds{1}_{\mathcal{N} \leq N}\left(  \phi_+ ( f_{(t_k;0)}) -  \widetilde{\phi}_+ (f_{(t_k;0)} ) \right) \mathds{1}_{\mathcal{N} \leq N} \prod_{j=k+1}^m   \big( \mathds{1}_{\mathcal{N} \leq N} \widetilde{\phi}_+ (f_{(t_i;0)} )  \big)   \Omega \rangle \; . 
\end{align}
With \eqref{eq:a-b} we find 
\begin{align}
\| &\left(  \phi_+ ( f_{(t_k;0)}) -  \widetilde{\phi}_+ (f_{(t_k;0)} ) \right)  \prod_{j=k+1}^m  \big( \mathds{1}_{\mathcal{N} \leq N} \widetilde{\phi}_+ (f_{(t_j;0)} ) \mathds{1}_{\mathcal{N} \leq N} \big) \Omega  \|   \notag \\
&\quad \leq  CN^{-1/2} e^{C \vert t_k \vert} \| ( \mathcal{N} + 1)^{3/2} \prod_{j=k+1}^m  \big( \mathds{1}_{\mathcal{N} \leq N} \widetilde{\phi}_+ (f_{(t_j;0)} ) \mathds{1}_{\mathcal{N} \leq N} \big) \Omega \| \; . 
\end{align}
We remark that \eqref{eq:N-phi+} holds true for $\widetilde{\phi}_+$, too, as the commutation relation of $a^*(h),a(h)$ with $\mathcal{N}$ are similar to the commutation relations of  $b^*(h),b(h)$ with $\mathcal{N}$. Thus we find that 
\begin{align}
\| &\left(  \phi_+ ( f_{(t_k;0)}) -  \widetilde{\phi}_+ (f_{(t_k;0)} ) \right)  \prod_{j=k+1}^m  \big( \mathds{1}_{\mathcal{N} \leq N} \widetilde{\phi}_+ (f_{(t_j;0)} ) \mathds{1}_{\mathcal{N} \leq N} \big)  \Omega  \|  \leq  C_m N^{-1/2} e^{C  \sum_{j=k}^m \vert t_j \vert}  
\end{align}
and again from \eqref{eq:N-phi+} 
\begin{align}
\| \prod_{i=1}^{k-1}    \phi_+ ( f_{(t_i;0)}  )\Omega \| \leq C_k e^{C  \sum_{j=1}^{k-1} \vert t_j \vert}  \;  
\end{align}
that all together leads to 
\begin{align}
\vert \langle \Omega, & \prod_{i=1}^m   \phi_+ ( f_{(t_i;0)} )   \Omega \rangle -\langle \Omega,  \prod_{i=1}^m  \big( \mathds{1}_{\mathcal{N} \leq N} \widetilde{\phi}_+ (f_{(t_i;0)} ) \mathds{1}_{\mathcal{N} \leq N} \big)  \Omega \rangle \vert \leq C_m N^{-1/2} e^{C  \sum_{j=1}^{m} \vert t_j \vert}\label{eq:step311}
\end{align}
For the second line of the r.h.s. of \eqref{eq:step301} we find 
\begin{align}
\langle \Omega, & \prod_{i=1}^m  \big( \mathds{1}_{\mathcal{N} \leq N} \widetilde{\phi}_+ (f_{(t_i;0)} ) \mathds{1}_{\mathcal{N} \leq N} \big)  \Omega \rangle -\langle \Omega,  \prod_{i=1}^m   \widetilde{\phi}_+ ( f_{(t_i;0)} )   \Omega \rangle  \notag \\
=& \sum_{k=1}^m \langle \Omega,  \prod_{i=1}^{k-1}    \widetilde{\phi}_+ ( f_{(t_i;0)} ) \left( \mathds{1}_{\mathcal{N} \leq N}  \widetilde{\phi}_+ ( f_{(t_k;0)}) \mathds{1}_{\mathcal{N} \leq N} -  \widetilde{\phi}_+ (f_{(t_k;0)} ) \right)  \prod_{j=k+1}^m   \big( \mathds{1}_{\mathcal{N} \leq N} \widetilde{\phi}_+ (f_{(t_i;0)} ) \mathds{1}_{\mathcal{N} \leq N} \big)   \Omega \rangle \; . 
\end{align}
we find with \eqref{eq:diff-tildephi} and similar arguments as before 
\begin{align}
    \big\vert \langle \Omega, & \prod_{i=1}^m  \big( \mathds{1}_{\mathcal{N} \leq N} \widetilde{\phi}_+ (f_{(t_i;0)} ) \mathds{1}_{\mathcal{N} \leq N} \big)  \Omega \rangle -\langle \Omega,  \prod_{i=1}^m   \widetilde{\phi}_+ ( f_{(t_i;0)} )   \Omega \rangle \big\vert 
    \leq  C_m N^{-1} e^{C  \sum_{j=1}^{m} \vert t_j \vert}
\end{align}
yielding with \eqref{eq:step311} to 
\begin{align}
    \vert \langle \Omega, & \prod_{i=1}^m   \phi_+ ( f_{(t_i;0)} )   \Omega \rangle -\langle \Omega,  \prod_{i=1}^m   \widetilde{\phi}_+ (f_{(t_i;0)} )   \Omega \rangle \vert \leq C_m N^{-1/2} e^{C  \sum_{j=1}^{m} \vert t_j \vert} \; . 
\end{align}

\subsection*{Step 4} We are left with computing 
\begin{align}
\langle \Omega,  \prod_{i=1}^m  \widetilde{\phi}_+ (f_{(t_i;0)} ) \Omega \rangle = \langle \Omega,  \prod_{i=1}^m  \left( a^* (f_{(t_i;0)} ) + a (f_{(t_i;0)} ) \right) \Omega \rangle \; . 
\end{align}
First we observe that for odd $m$ this is a sum of expectation values each of an odd number of creation and annihilation operators and that, thus, all vanish and \eqref{eq:odd}. Second, assume that $m$ is even. With the commutation relations we find 
\begin{align}
\langle \Omega,  \prod_{i=1}^m  \widetilde{\phi}_+ (f_{(t_i;0)} ) \Omega \rangle  =& \langle \Omega,  a (f_{(t_1;0)}) \prod_{i=2}^{m-1}  \widetilde{\phi}_+ (f_{(t_i;0)} ) a^*( f_{(t_m;0)} ) \Omega \rangle \notag \\
=& \sum_{ \pi \in \Pi_m} ( \Sigma_{t_1, \dots, t_m})_{\pi_1, \pi_2} \cdots ( \Sigma_{t_1, \dots, t_m})_{\pi_m, \pi_{m-1}}
\end{align}
where the sum runs over the set $\Pi_m :=  \lbrace  \pi \in \mathcal{S}_m \vert \pi_{2i -1} < \pi_{2i} \; \quad \text{for all} \; 1 \leq i \leq m/2 \quad \text{and} \; \pi_{2i-1} < \pi_{2i+1} \quad \text{for all} \; 1 \leq i \leq (m-1)/2 \rbrace $, i.e. over all permutations $\pi = ( \pi_1, \dots \pi_m)$ of the set $\lbrace 1, \dots, m \rbrace$ such that $\pi_{2i -1} < \pi_{2i}$ for all $1 \leq i \leq m/2 $ and $\pi_{2i-1} < \pi_{2i+1}$ for all $1 \leq i \leq (m-1)/2$ and the constants $(\Sigma_{t_1, \dots, t_m})_{i,j}$ are given by \eqref{def:Sigma}.\end{proof}

\section{Proof of Theorem \ref{thm:OTOC}} 
\label{sec:OTOC}

As explained in the introduction we use similar techniques as in the proof of the other two Theorems \ref{thm:MCLT} and Theorem \ref{thm:corrfct} in Sections \ref{sec:MCLT} resp. Section \ref{sec:corr} to prove Theorem \ref{thm:OTOC} in this section.

\begin{proof}[Proof of Theorem \ref{thm:OTOC}]
We recall that by definition \eqref{def:Atj} 
\begin{align}
A_t^{(j)} = e^{iH_Nt} \left( A^{(j)} - \langle \varphi_t, A \varphi_t \rangle \right) e^{-iH_N t}
\end{align}
and start with the observation that by symmetry of $\psi_N$ we have 
\begin{align}
 \langle \psi_N, &  (i[A_t^{(1)} , (N-1)B^{(2)}_0 + B_0^{(1)} ])^2 \psi_N \rangle \notag \\
=& -\langle \psi_N,  \Big( (N-1)^2  ([ A_t^{(1)}, B^{(2)}_0 ] )^2+ ([ A_t^{(1)}, B^{(1)}_0 ] )^2 ] \Big) \psi_N \rangle  \notag \\
 &- \langle \psi_N,  \Big(   (N-1) [ A_t^{(1)}, B^{(2)}_0 ][ A_t^{(1)}, B^{(1)}_0 ] - (N-1) [ A_t^{(1)}, B^{(1)}_0 ][ A_t^{(1)}, B^{(2)}_0 ]\Big) \psi_N \rangle \notag \\
=& \frac{(N-1)^2}{N^2(N-1)^2} \sum_{i\not= j, k \not=\ell }^N \langle \psi_N, [A^{(j)}_t, B^{(i)}_0] \; [B^{(k)}_0, A^{(\ell)}_t] \psi_N \rangle\\
&+ \frac{1}{N^2} \sum_{i,k=1 }^N \langle \psi_N,  [A^{(i)}_t, B^{(i)}_0 ] \;  [B^{(k)}_0, A^{(k)}_t ])^2 \psi_N \rangle \notag \\
&+ \frac{(N-1)}{N^2(N-1)} \sum_{i\not= j, k,\ell =1 }^N \langle \psi_N, [A^{(j)}_t, B^{(i)}_0] \; [B^{(k)}_0, A^{(\ell)}_t] \psi_N \rangle \notag \\
&+ \frac{(N-1)}{N^2(N-1)} \sum_{k\not= \ell, i,j=1 }^N \langle \psi_N, [A^{(j)}_t, B^{(i)}_0] \; [B^{(k)}_0, A^{(\ell)}_t] \psi_N \rangle  \notag \\
=&\frac{1}{N^2} \sum_{i,j,k, \ell =1}^N  \langle \psi_N, (i [A^{(j)}_t, B^{(i)}_0]\;  [B_0^{(k)}, A_t^{(\ell)}] \psi_N \rangle  \notag \\
=& \frac{1}{N^2}\langle \psi_N, [d\Gamma (A_t), d\Gamma (B_0)] \; [d\Gamma (B_0), d\Gamma (A_t)]\psi_N \rangle \; . \label{eq:square}
\end{align}
Thus the goal in the following is to approximately compute expectation values of the form
\begin{align}
\frac{1}{N^2} \langle \psi_N, d \Gamma (A_{1,t_1}) d\Gamma (B_{t_2}) d \Gamma (D_{t_3}) d\Gamma (E_{t_4}) \psi_N \rangle \; 
\end{align}
where $A,B,D,E: L^2( \mathbb{R}^3) \rightarrow L^2( \mathbb{R}^3)$ denote self-adjoint operators satisfying \eqref{ass:A} and the operators $A_{t_1}, B_{t_2}, D_{t_3}, E_{t_4}$ are defined as in \eqref{def:Atj} for $t_i \in \mathbb{R}$. (Notice that at this point it would be easier to use the functoriality property $[d\Gamma(A),d\Gamma(B)]=d\Gamma([A,B])$, but we partly view the proof as warm-up for the proof of the more general Theorem \ref{thm:corrfct} later on. For this, the following calculation is more instructive.)
For this we define 
\begin{align}
f^{(1)}_{(t,0)} := \mathcal{L}_{(t,0)} q_t A \varphi_t, \quad f^{(2)}_{(t;0)} := \mathcal{L}_{(t,0)} q_t B \varphi_t, \quad f^{(3)}_{(t;0)} := \mathcal{L}_{(t,0)} q_t D \varphi_t, \quad f^{(4)}_{(t;0)} := \mathcal{L}_{(t,0)} q_t E \varphi_t
\end{align}
with $\mathcal{L}_{(t;0)}$ given by \eqref{def:Ls} and find from 
Theorem \ref{thm:corrfct} 
\begin{align}\label{eq:rateref}
\Big \vert &\langle \psi_N,  N^{-2} d \Gamma (A_{t_1}) d\Gamma (B_{t_2}) d \Gamma (D_{t_3}) d\Gamma (E_{t_4})\psi_N \rangle \notag \\ 
&\quad    -\langle f_{(t_1,0)}, f_{(t_4,0)} \rangle  \langle f_{(t_2,0)}, f_{(t_3,0)} \rangle  - \langle f_{(t_1,0)}, f_{(t_2,0)} \rangle  \langle f_{(t_3,0)}, f_{(t_4,0)} \rangle \notag \\
&	\hspace{7cm} -  \langle f_{(t_1,0)}, f_{(t_3,0)} \rangle  \langle f_{(t_2,0)}, f_{(t_4,0)} \rangle  \Big\vert  \notag \\
&\leq C N^{-1/2}  e^{e^{ C( \vert t_1 \vert + \vert t_2 \vert + \vert t_3 \vert + \vert t_4 \vert) }}  \; . 
\end{align}
In particular we get from \eqref{eq:square}
\begin{align}
\label{eq:rate2}
 \langle \psi_N, & i \Big(   (i[ A_t^{(1)}, (N-1) B^{(2)}_0+ B_0^{(1)} ])^2 \Big) \psi_N \rangle \notag\\
 =&  \langle \psi_N, [d\Gamma (A_t), d\Gamma (B_0)] \; [d\Gamma (B_0), d\Gamma (A_t)]\psi_N \rangle \; \notag \\
=& \langle \psi_N, \left[ d\Gamma (A_t) d\Gamma (B_0) d\Gamma (B_0) d\Gamma (A_t)  - d\Gamma (A_s) d\Gamma (A_t) d\Gamma (A_s) d\Gamma (A_t) \right]  \psi_N \rangle \notag \\
&- \langle \psi_N, \left[ d\Gamma (A_t) d\Gamma (B_0) d\Gamma (A_t) d\Gamma (B_0)  - d\Gamma (B_0) d\Gamma (A_t) d\Gamma (A_t) d\Gamma (B_0) \right]  \psi_N \rangle  
\end{align}
and thus denoting in abuse of notation we set 
\begin{align}
\label{def:f,g}
g_0 :=  q_0 B \varphi_0, \quad f_t := f_{(t,0)} = \mathcal{L}_{(t;0)} q_t O \varphi_t 
\end{align}
and find 
\begin{align}
 \lim_{N \rightarrow \infty} \langle \psi_N, & i \Big(   (i[ A_t^{(1)}, (N-1) B^{(2)}_0+ B_0^{(1)} ])^2 \Big) \psi_N \rangle \notag\\
 =&  \langle f_{t}, g_{0} \rangle  \langle g_{0}, f_{t} \rangle + \| f_t \|^2 \| g_0 \|^2 + \langle f_{t}, g_{0} \rangle  \langle g_{0}, f_{t} \rangle  \notag \\
  &-  \langle g_{0}, f_{t} \rangle  \langle g_{0}, f_{t}\rangle -  \langle g_{0}, f_{t} \rangle  \langle f_{t}, g_{0} \rangle - \| g_0 \|^2 \| f_t \|^2 \notag \\
 &-  \langle f_{t}, g_{0} \rangle  \langle f_{t}, g_{0} \rangle  - \langle f_{t}, g_{0}\rangle \langle g_{0}, f_{t}\rangle   - \| g_0 \|^2 \| f_t \|^2\notag \\
 &+ \langle g_{0}, f_{t} \rangle  \langle f_{t}, g_{0} \rangle + \langle f_{t}, g_{0} \rangle  \langle g_{0}, f_{t} \rangle + \| f_t \|^2 \| g_0 \|^2\notag \\
 =& 2  \vert  \langle g_{0}, f_{t} \rangle \vert^2 - 2 \Re \langle g_{0}, f_{t} \rangle^2 \notag \\
 =& 2 (\Im \langle g_{0}, f_{t} \rangle )^2   \; . 
\end{align}
Thus, in order to conclude with Theorem \ref{thm:OTOC} it remains to show that  
\begin{align}
\label{eq:U,V-f}
\Im \langle g_{0}, f_{t} \rangle = -\frac{i}{2} \langle ( \varphi_0,  J \varphi_0),  S \left[ B ,  \Theta (t;0) A\Theta (t;0)^{-1} \right] (\varphi_0,  J  \varphi_0) \rangle_{L^2 \oplus L^2 }  
\end{align}
where $\Theta(t;s)$ is given by \eqref{def:Theta}. To this end we observe that by definition of $g_0, f_t$ in \eqref{def:f,g} we have 
\begin{align}
\label{eq:refer}
\Im \langle g_{0}, f_{t} \rangle  = \langle q_0 B \varphi_0,  \mathcal{L}_{(t,0)}    \; A q_t \varphi_t \rangle 
\end{align} 
 where $\mathcal{L}_{(t;0)}$ satisfies the first identity of \eqref{def:Ls}. As a first step we write 
\begin{align}
 2 \Im \langle g_{0}, f_{(t;0)} \rangle 
 =&2 \Im \langle B \varphi_0, \mathcal{L}_{(t,0)}  A \varphi_t \rangle  \notag \\
 &- 2  \langle \varphi_0, B \varphi_0 \rangle \; \Im \langle \varphi_0,\mathcal{L}_{(t,0)}    A \varphi_t \rangle \notag \\
 &- 2 \langle \varphi_t, A \varphi_t \rangle \; \Im \langle B \varphi_0, \mathcal{L}_{(t,0)}  \varphi_t \rangle \notag \\
 & +  \langle \varphi_0, B \varphi_0 \rangle \;\langle \varphi_t, A \varphi_t \rangle \;  \Im \langle \varphi_0, \mathcal{L}_{(t),0}  \varphi_t \rangle 
 \label{eq:splitting-Im}
\end{align}
and the last three lines vanish as we explain in the following. To see this we write the vector $ (\mathcal{L}_{(t,0)}   A\varphi_t, J \mathcal{L}_{(t,0)}  A \varphi_t)  \in L^2( \mathbb{R}^3) \oplus L^2( \mathbb{R}^3)$ as 
\begin{align}
\label{eq:fst-anders}
 (\mathcal{L}_{(t,0)}  A\varphi_t, J \mathcal{L}_{(t,0)}   A \varphi_t) = \Theta(t;0) ( A\varphi_t, J A \varphi_t) \; .
\end{align}
Next we use this formalism to show that the last three lines of the r.h.s. of \eqref{eq:splitting-Im} vanish. For the second term of \eqref{eq:splitting-Im} we write 
\begin{align}
2 \Im \langle \varphi_0, \mathcal{L}_{(t;0)} A \varphi_t \rangle =& -i \langle ( \varphi_0, - J \varphi_0) , \;  ( \mathcal{L}_{(t;0)}A \varphi_t, J \mathcal{L}_{(t;0)} A \varphi_t ) \rangle_{L^2 \oplus L^2 } \notag \\
=& -i \langle ( \varphi_0, - J \varphi_0) , \;  \Theta (t;0) (A \varphi_t, J A \varphi_t ) \rangle_{L^2 \oplus L^2 }
\end{align}
that with \eqref{eq:prop-Theta} becomes 
\begin{align}
2 \Im \langle \varphi_0, \mathcal{L}_{(t;0)}  A \varphi_t \rangle =& -i  \langle \Theta (t;0)^* ( \varphi_0, - J \varphi_0) , \;   (A \varphi_t, JA \varphi_t ) \rangle_{L^2 \oplus L^2 } \notag \\
=& -i\langle \Theta (t;0)^{-1}   ( \varphi_0, J \varphi_0) , \;   (A \varphi_t, - JA \varphi_t ) \rangle_{L^2 \oplus L^2 } \notag \\
=& -i  \langle ( \varphi_t,  J \varphi_t) , \;   (A \varphi_t, -JA \varphi_t ) \rangle_{L^2 \oplus L^2 } \notag \\
=&  2 \Im \langle \varphi_t, A \varphi_t \rangle \notag \\
=& 0 \; \label{eq:term1}
\end{align}
where we concluded by \eqref{eq:prop-bogo}. For the second term we proceed analogously and calculate 
\begin{align}
2 \Im \langle B \varphi_0, \mathcal{L}_{(t;0)} \varphi_t \rangle =& -i  \langle   ( B\varphi_0, - JB \varphi_0) , \;  \Theta (t;0) ( \varphi_t, J \varphi_t ) \rangle_{L^2 \oplus L^2 } \notag \\
=&-i  \langle   ( B\varphi_0, - J B\varphi_0) , \;   (\varphi_0, J \varphi_0 ) \rangle_{L^2 \oplus L^2 } \notag \\
=& 2 \Im \langle \varphi_0, B \varphi_0 \rangle \notag \\
=& 0 \; . \label{eq:term2}
\end{align} 
and finally for the third term 
\begin{align}
2 \Im \langle  \varphi_0, \mathcal{L}_{(t;0)} \varphi_t \rangle =& -i  \langle   ( \varphi_0, - J \varphi_0) , \;  \Theta (t;0) ( \varphi_t, J \varphi_t ) \rangle_{L^2 \oplus L^2 } \notag \\
=&-i  \langle   ( \varphi_0, - J \varphi_0) , \;   ( \varphi_0, J \varphi_0 ) \rangle_{L^2 \oplus L^2 } \notag \\
=& 2 \Im \langle \varphi_0,  \varphi_0 \rangle \notag \\
=& 0\; .\label{eq:term3}
\end{align} 
Summarizing \eqref{eq:term1}, \eqref{eq:term2} and \eqref{eq:term3} we find from \eqref{eq:splitting-Im} that we are left with calculating 
\begin{align}
\label{eq:Im-1}
2 \Im \langle g_{0}, f_{t} \rangle =& \Im \langle B \varphi_0, \mathcal{L}_{(t;0)} A \varphi_t \rangle \notag \\
=&  -i  \langle ( \varphi_0, J\varphi_0),   S B  \Theta (t;0) A(\varphi_t, J \varphi_t) \rangle_{L^2 \oplus L^2 } \notag \\
=&  -i  \langle ( \varphi_0,  J \varphi_0),  S B  \Theta (t;0) A\Theta (t;0)^{-1} (\varphi_0,  J  \varphi_0) \rangle_{L^2 \oplus L^2 } 
\end{align}
where we used that $A,B$ are real and self-adjoint operators. As the right and left hand side are both real, we can take the real part of both sides and obtain 
\begin{align}
2 \Im \langle g_{0}, & f_{t} \rangle \notag \\
=&    \frac{i}{2} \langle ( \varphi_0,  J \varphi_0), \left( S B  \Theta (t;0) A\Theta (t;0)^{-1} - ( \Theta (t;0)^{-1} )^* A \Theta (t;0)^* B S \right) (\varphi_0,  J  \varphi_0) \rangle_{L^2 \oplus L^2 } 
\end{align}
From \eqref{eq:prop-Theta} we have $( \Theta (t;0)^{-1} )^* = S \Theta (t;0) S$ and $\Theta (t;s)^* = S \Theta (t;0)^{-1} S$ and we conclude by 
\begin{align}
2 \Im \langle g_{0},  f_{t} \rangle =   -\frac{i}{2} \langle ( \varphi_0,  J \varphi_0),  S \left[ B ,  \Theta (t;0) A\Theta (t;0)^{-1} \right] (\varphi_0,  J  \varphi_0) \rangle_{L^2 \oplus L^2 } 
\end{align}
that is the r.h.s. of \eqref{eq:thm-OTOC}.
\end{proof}

\section{Proof of Corollary \ref{cor:OTOC}} 
\label{sec:cor}

The proof of Corollary \ref{cor:OTOC} shown here in this Section is based on Theorem \ref{thm:OTOC} and the properties of $\widetilde{K}_{j,t}$ and $\mathcal{L}_{(t,s)}$ from Lemmas \ref{lemma:K} resp. Lemma \ref{lemma:f}.  

\begin{proof}[Proof of Corollary \ref{cor:OTOC}] We first prove part (ii) and part (i) afterwards. 

\subsection*{Proof of (ii)} The bound follows immediately from Lemma \ref{lemma:f}. 

\subsection*{Proof of (i)} We are interested in an expansion for small times $\vert t \vert \leq T$ of (see the first line of \ref{eq:Im-1})
\begin{align}
\Im \langle \varphi_0, B \mathcal{L}_{(t;0)} A  \varphi_t \rangle  = \Im \langle \varphi_0, B \mathcal{L}_{(t;0)} A \mathcal{L}_{(0,t)} \varphi_0 \rangle 
\end{align}
where the last identity follows from \eqref{def:Lt}. We write with Duhamel's formula 
 \begin{align}
\Im \langle \varphi_0, B \mathcal{L}_{(t;0)} A \mathcal{L}_{(0;t)} \varphi_0 \rangle  =  \Im \langle \varphi_0, B A \varphi_0 \rangle - \int_0^t ds \frac{d}{ds} \Im \langle \varphi_0, B \mathcal{L}_{(t,s)} A \mathcal{L}_{(s;t)} \varphi_0 \rangle \; . 
 \end{align}
With \eqref{def:Ls} we find 
 \begin{align}
 \Im \langle \varphi_0, B \mathcal{L}_{(t;0)} A \mathcal{L}_{(0;t)} \varphi_0 \rangle  =&  \Im \langle \varphi_0, B A \varphi_0 \rangle \notag \\
 &+  \int_0^t ds  \Re \langle \varphi_0, B \left[  h_{\varphi_s} + \widetilde{K}_{1,s} - \widetilde{K}_{2,s} J, \;  \mathcal{L}_{(t,s)} A \mathcal{L}_{(s;t)}  \right] \varphi_0 \rangle \notag \\ 
 \end{align}
that we can further write as 
 \begin{align}
 \Im & \langle \varphi_0, B \mathcal{L}_{(t;0)} A \mathcal{L}_{(0;t)} \varphi_0 \rangle \notag \\
  =&  \Im \langle \varphi_0, B A \varphi_0 \rangle - t \Re \langle \varphi_0, B \left[  h_{\varphi_t} + \widetilde{K}_{1,t} - \widetilde{K}_{2,t} J, \;   A  \right] \varphi_0 \rangle  \notag \\ 
 &-  \int_0^t \int_s^t ds d\tau  \; \Re \langle \varphi_0, B \left[  \dot{h}_{\varphi_\tau} + \dot{\widetilde{K}}_{1,\tau} - \dot{\widetilde{K}}_{2,\tau} J, \;  \mathcal{L}_{(t,\tau)} A \mathcal{L}_{(\tau;t)} \right] \varphi_0 \rangle \notag \\ 
  &+  \int_0^t \int_s^t ds d\tau  \; \Im \langle \varphi_0, B \left[h_{\varphi_\tau} + \widetilde{K}_{1,\tau} - \widetilde{K}_{2,\tau} J, \;  \left[ h_{\varphi_\tau} + \widetilde{K}_{1,\tau} - \widetilde{K}_{2,\tau} J, \;  \mathcal{L}_{(t,\tau)} A \mathcal{L}_{(\tau;t)}  \right]\right] \varphi_0 \rangle \; . 
 \end{align}
We apply once more Duhamel's formula for the second term of the r.h.s. and arrive at 
 \begin{align}
 \Im & \langle \varphi_0, B \mathcal{L}_{(t;0)} A \mathcal{L}_{(0;t)} \varphi_0 \rangle \notag \\
 =&  \Im \langle \varphi_0, B A \varphi_0 \rangle - t \Re \langle \varphi_0, B \left[  h_{\varphi_0} + \widetilde{K}_{1,0} - \widetilde{K}_{2,0} J, \;   A  \right] \varphi_0 \rangle  \notag \\ 
  &- t \int_0^t ds \; \Re \langle \varphi_0, B \left[  \dot{h}_{\varphi_s} + \dot{\widetilde{K}}_{1,s} - \dot{\widetilde{K}}_{2,s} J, \;   A  \right] \varphi_0 \rangle \notag \\
 &-  \int_0^t \int_s^t ds d\tau  \; \Re \langle \varphi_0, B \left[  \dot{h}_{\varphi_\tau} + \dot{\widetilde{K}}_{1,\tau} - \dot{\widetilde{K}}_{2,\tau} J, \;  \mathcal{L}_{(t,\tau)} A \mathcal{L}_{(\tau;t)} \right] \varphi_0 \rangle \notag \\ 
  &+  \int_0^t \int_s^t ds d\tau  \; \Im \langle \varphi_0, B \left[h_{\varphi_\tau} + \widetilde{K}_{1,\tau} - \widetilde{K}_{2,\tau} J, \;  \left[ h_{\varphi_\tau} + \widetilde{K}_{1,\tau} - \widetilde{K}_{2,\tau} J, \;  \mathcal{L}_{(t,\tau)} A \mathcal{L}_{(\tau;t)}  \right]\right] \varphi_0 \rangle \; . \label{eq:Duhamel}
 \end{align}
We show in the following that the last three lines of the r.h.s. of \eqref{eq:Duhamel} are $O(t^2)$ and thus are sub-leading for all $t \ll 1$. For the term in the second line of the r.h.s. \eqref{eq:Duhamel} we find 
\begin{align}
 t & \int_0^t ds \;  \vert \langle \varphi_0, B \left[  \dot{h}_{\varphi_s} + \dot{\widetilde{K}}_{1,s} - \dot{\widetilde{K}}_{2,s} J, \;   A  \right] \varphi_0 \rangle \vert \notag \\
 & \leq t \int_0^t ds \| B \varphi_0 \|_2 \left( \| \dot{h}_{\varphi_s} A \varphi_0 \|_2 + \| \dot{\widetilde{K}}_{2,s}J A \varphi_0 \|_2 + \| \dot{\widetilde{K}}_{1,s} A \varphi_0 \|_2 \right) \; \notag \\
 &+ t \int_0^t ds \| A B \varphi_0 \|_2 \left( \| \dot{h}_{\varphi_s} \varphi_0 \|_2 + \| \dot{\widetilde{K}}_{2,s}J \varphi_0 \|_2 + \| \dot{\widetilde{K}}_{1,s}\varphi_0 \|_2 \right) \; . 
\end{align}
From Lemma \ref{lemma:K} we have 
\begin{align}
\| \dot{h}_{\varphi_s} A \varphi_0 \|_2 \leq e^{ C s} \| A \varphi_0 \|_2 \leq C \| A \varphi_0 \|_2 , \quad \| \dot{h}_{\varphi_s}  \varphi_0 \|_2 \leq e^{ C s} \|  \varphi_0 \|_2 \leq C \|  \varphi_0 \|_2 
\end{align}
for all $0 \leq s \leq t \leq 1$ and similarly 
\begin{align}
\| \dot{\widetilde{K}}_{j,s} A \varphi_0 \|_2 \leq e^{ C s} \| A \varphi_0 \|_2 \leq C \| A \varphi_0 \|_2 , \quad \| \dot{\widetilde{K}}_{j,s} \varphi_0 \|_2 \leq e^{ C s} \|  \varphi_0 \|_2 \leq C \|  \varphi_0 \|_2 \; 
\end{align}
so that we arrive at 
\begin{align}
 t & \int_0^t ds \;  \vert \langle \varphi_0, B \left[  \dot{h}_{\varphi_s} + \dot{\widetilde{K}}_{1,s} + \dot{\widetilde{K}}_{2,s} J, \;   A  \right] \varphi_0 \rangle \vert \notag \\
 &\leq  C \left( \| A B \varphi_0 \|_2 + \| A \varphi_0 \|_2 \; \| B \varphi_0 \|_2 \right) t  \int_0^t ds \leq C  \left( \| A B \varphi_0 \|_2 + \| A \varphi_0 \|_2 \; \| B \varphi_0 \|_2\right)  t^2 \; . \label{eq:estimate1}
\end{align}
For the term in the third line of the r.h.s. of \eqref{eq:Duhamel} we proceed with similar ideas and find 
\begin{align}
 \int_0^t & \int_s^t ds d\tau  \; \vert \langle \varphi_0, B \left[  \dot{h}_{\varphi_\tau} + \dot{\widetilde{K}}_{1,\tau} - \dot{\widetilde{K}}_{2,\tau} J, \;  \mathcal{L}_{(t,\tau)} A \mathcal{L}_{(\tau;t)}  \right] \varphi_0 \rangle \vert \notag \\
 \leq& \int_0^t \int_s^t ds d\tau \;  \|B \varphi_0 \|_2  \notag \\
 & \hspace{1cm} \times \left( \| \dot{h}_{\varphi_\tau}  \mathcal{L}_{(t,\tau)} A \mathcal{L}_{(\tau;t)}  \varphi_0 \|_2  +\| \dot{\widetilde{K}}_{1,\tau}  \mathcal{L}_{(t,\tau)} A \mathcal{L}_{(\tau;t)}  \varphi_0 \|_2 + \| \dot{\widetilde{K}}_{2,\tau} J  \mathcal{L}_{(t,\tau)} A \mathcal{L}_{(t;\tau)}^-  \varphi_0 \|_2 \right) \notag \\
 &+  \int_0^t \int_s^t ds d\tau \;  \|B \varphi_0 \|_2 \notag \\
 \leq& \int_0^t \int_s^t ds d\tau \;  \|B \varphi_0 \|_2  \notag \\
 & \hspace{1cm} \times \left( \|  \mathcal{L}_{(t,\tau)} A \mathcal{L}_{(\tau;t)} \dot{h}_{\varphi_\tau}  \varphi_0 \|_2  +\|  \mathcal{L}_{(t,\tau)} A \mathcal{L}_{(\tau;t)}  \dot{\widetilde{K}}_{1,\tau}  \varphi_0 \|_2 + \|  \mathcal{L}_{(t,\tau)} A \mathcal{L}_{(\tau;t)} \dot{\widetilde{K}}_{2,\tau} J  \varphi_0 \|_2 \right) 
\end{align}
Since by Lemma \ref{lemma:f} 
\begin{align}
\label{eq:estiamte-LL}
\| \mathcal{L}_{(t,\tau)} A \mathcal{L}_{(\tau;t)} g \|_2 \leq C e^{ c (t-\tau)} \| A \mathcal{L}_{(\tau;t)}  g \|_2  \leq C e^{ e^{c (t-\tau)}} \| \mathcal{L}_{(\tau;t)}  g \|_{H^2} \leq C  \| g \|_{H^2} 
\end{align}
for all $0 \leq \tau \leq t$ and any $g \in H^2$, we conclude with 
\begin{align}
 \int_0^t & \int_s^t ds d\tau  \; \vert \langle \varphi_0, B \left[  \dot{h}_{\varphi_\tau} + \dot{\widetilde{K}}_{1,\tau} + \dot{\widetilde{K}}_{2,\tau} J, \;  \mathcal{L}_{(t,\tau)} A \mathcal{L}_{(\tau;t)}  \right] \varphi_0 \rangle \vert \leq C  \int_0^t  \int_s^t ds d\tau \leq C t^2 \; . \label{eq:estimate2}
\end{align}
For the remaining forth term of the r.h.s. of \eqref{eq:splitting-Im} we find 
\begin{align}
 \int_0^t  & \int_s^t ds d\tau  \; \Im \langle \varphi_0, B \left[h_{\varphi_\tau} + \widetilde{K}_{1,\tau} - \widetilde{K}_{2,\tau} J, \;  \left[ h_{\varphi_\tau} + \widetilde{K}_{1,\tau} - \widetilde{K}_{2,\tau} J, \;  \mathcal{L}_{(t,\tau)} A \mathcal{L}_{(\tau;t)}  \right]\right] \varphi_0 \rangle  \notag \\
\leq& \int_0^t \int_s^t ds d\tau   \|  (h_{\varphi_\tau} + \widetilde{K}_{1,\tau} - \widetilde{K}_{2,\tau} J  ) B \varphi_0 \|_2 \| ( h_{\varphi_\tau} + \widetilde{K}_{1,\tau} - \widetilde{K}_{2,\tau} J ) \mathcal{L}_{(t,\tau)} A \mathcal{L}_{(\tau;t)} \varphi_0 \|_2 \notag \\
&+  2 \int_0^t \int_s^t ds d\tau   \| B \varphi_0 \|_2 \| ( h_{\varphi_\tau} + \widetilde{K}_{1,\tau} - \widetilde{K}_{2,\tau} J ) \mathcal{L}_{(t,\tau)} A \mathcal{L}_{(\tau;t)}  (h_{\varphi_\tau} + \widetilde{K}_{1,\tau} - \widetilde{K}_{2,\tau} J  ) \varphi_0 \|_2 \notag \\
  &+  \int_0^t \int_s^t ds d\tau   \| B \varphi_0 \|_2 \| \mathcal{L}_{(t,\tau)} A \mathcal{L}_{(\tau;t)}  (h_{\varphi_\tau} + \widetilde{K}_{1,\tau} - \widetilde{K}_{2,\tau} J  )^2 \varphi_0 \|_2
\end{align}
and we conclude with Lemma \ref{lemma:f} similarly as in \eqref{eq:estiamte-LL} that 
\begin{align}
\label{eq:estimate3}
 \int_0^t  & \int_s^t ds d\tau  \; \vert  \Im \langle \varphi_0, B \left[h_{\varphi_\tau} + \widetilde{K}_{1,\tau} - \widetilde{K}_{2,\tau} J, \;  \left[ h_{\varphi_\tau} + \widetilde{K}_{1,\tau} - \widetilde{K}_{2,\tau} J, \;  \mathcal{L}_{(t,\tau)} A \mathcal{L}_{(\tau;t)}  \right]\right] \varphi_0 \rangle  \vert \notag \\
 & \hspace{2cm} \leq C t^2. 
\end{align}
Summarizing \eqref{eq:estimate1}, \eqref{eq:estimate2} and \eqref{eq:estimate3} we thus get
\begin{align}
 \Im & \langle \varphi_0, B \mathcal{L}_{(t;0)} A \mathcal{L}_{(0;t)} \varphi_0 \rangle  \notag \\
 &= \Im \langle \varphi_0, B A \varphi_0 \rangle - t \Re \langle \varphi_0, B \left[  h_{\varphi_0} + \widetilde{K}_{1,0} - \widetilde{K}_{2,0} J, \;   A  \right] \varphi_0 \rangle + O(t^2) \; . 
\end{align}
\end{proof}

\section{Conclusions and future directions}
In this work, we consider bosons in mean-field scaling. This is a truly interacting model in which it is nonetheless possible to describe the particles quite precisely. Building on techniques developed in Bogoliubov theory for these systems (e.g., \cite{BKS,BSS,R}), we rigorously describe the large-$N$ asymptotics of out-of-time-ordered correlators for bosons by an effective nonlinear dynamics with symplectic structure. This is a rigorous many-body manifestation of the correspondence principle underlying quantum chaos. It can be understood in the context of the insight that the mean-field limit can be mathematically viewed as a many-body analog of a semiclassical limit; this observation goes back to Hepp \cite{H} and Ginibre-Velo \cite{GV}; see also \cite{RoS}. 

Our results open the door to studying many-body quantum information scrambling and entanglement generation in this paradigmatic interacting system. We give a few possibilities of interesting questions that can be tackled now, some of which we plan to return to in future work.

\begin{itemize}
    \item A central question in quantum chaos is to describe the intermediate-time growth of the OTOC, cf.\ \eqref{eq:Poisson} for the few-body version. Many-body analogs of this are extremely scarce, even on a non-rigorous level \cite{LSHOH,XS}. Our first main result can be seen as a many-body analog of the first identity in \eqref{eq:Poisson}. It would be extremely interesting to  also derive a version of the last $\approx $ in \eqref{eq:Poisson} (which involves the classical Lyapunov exponent), or more even modestly, a lower bound on the many-body OTOC that grows exponentially in time for intermediate times. Our result reduces this  many-body problem to a question about the  right-hand side of \eqref{eq:thm-OTOC}, i.e., a    question about the semiclassical growth behavior of a nonlinear dispersive PDE. This has two consequences: (i) our result opens up the vast toolkit of nonlinear dispersive PDE for quantum many-body chaos. (ii) our result makes numerical investigations of many-body information scrambling much more feasible, because the right-hand side of \eqref{eq:thm-OTOC} does not suffer anymore from the ``curse of dimensionality'' and can therefore be calculated quickly and precisely in many examples.
    
    \item One of the few topics that can rival quantum many-body chaos in terms of relevance and appeal in the wider physics community is many-body localization (MBL), which, can be expressed as slow growth of the OTOC \cite{CZHF,FZSZ}. Our result therefore also reduces the study of MBL mean-field bosons to a question about nonlinear dispersive PDE. More precisely, the results all go through in the presence of a random on-site external potential, which then manifests for the nonlinear Bogoliubov dynamics in the same way. Hence, it suffices to derive a nonlinear variant of Anderson localization as opposed to full MBL, which is actually conceivable from current techniques, see, e.g., \cite{CYZ}. We will return to this problem in a future paper with J.~Zhang.
    
    \item From a probabilistic perspective, various questions have been asked about dilute bosons in recent years: starting form the central limit theorem perspective of \cite{BKS,Rsing}, also large deviations \cite{RS} and Edgeworth expansions have been considered \cite{BP}. From a probabilistic perspective, another important extension of the CLT is  the pathwise functional CLT (or ``invariance principle'').  Our second main result, Theorem \ref{thm:MCLT}, is a stepping stone in this regard, as it can be formulated in terms of non-commutative probability to be identifying the so-called finite-dimensional distributions of this non-commutative time-indexed stochastic process with the natural Gaussian object given via the time-dependent  Wick rule. Making all this precise will require setting up a suitable free probability framework.
\end{itemize}

Finally, with an eye towards experimental realizations, it would be interesting to extend our results to the Gross-Pitaevskii regime which describes more realistic system of dilute bosons.

\section*{Acknowledgments}
The authors thank Christian Brennecke and Andreas Deuchert for helpful discussions. They also thank the referees for a number of suggestions that you improved manuscript. The research of ML is supported by  the Deutsche Forschungsgemeinschaft (DFG, German Research Foundation) through grant TRR 352-470903074.


\end{document}